\documentclass[10pt,a4paper]{article}

\usepackage{lmodern}
\usepackage[latin1]{inputenc}
\usepackage{graphicx,color}
\usepackage{epsfig}
\usepackage{verbatim}
\usepackage[english]{babel}
\usepackage{amsmath, amsthm, amssymb}
\usepackage{wrapfig}
\usepackage{neubauer}

\usepackage{placeins}
\usepackage{appendix}
\usepackage{csquotes}

\newtheorem{prop}{Proposition}
\newtheorem{theorem}{Theorem}

\def\dirac{\text{III}}

\DeclareMathOperator{\sinc}{sinc}

\makeatletter 
\renewcommand{\@fnsymbol}[1]{\@arabic{#1}}
\makeatother

\author{Victoria Hutterer\thanks{Industrial Mathematics Institute, Johannes Kepler University, Linz, Austria. victoria.hutterer@indmath.uni-linz.ac.at}, Ronny Ramlau\footnote{Industrial Mathematics Institute, Johannes Kepler University, Linz, and Johann Radon Institute for Computational and Applied Mathematics, Linz,  Austria.} \ and Iuliia Shatokhina\footnote{Johann Radon Institute for Computational and Applied Mathematics, Linz,  Austria.}}
\title{Real-time Adaptive Optics with pyramid wavefront sensors: Accurate wavefront reconstruction using iterative methods}

\begin{document}

\maketitle

\begin{abstract}
In this paper, we address the inverse problem of fast, stable, and high-quality wavefront reconstruction from pyramid wavefront sensor data for Adaptive Optics systems on Extremely Large Telescopes. For solving the indicated problem we apply well-known iterative mathematical algorithms, namely conjugate gradient, steepest descent, Landweber, Landweber-Kaczmarz and steepest descent-Kaczmarz iteration based on theoretical studies of the pyramid wavefront sensor. We compare the performance (in terms of correction quality and speed) of these algorithms in end-to-end numerical simulations of a closed adaptive loop. The comparison is performed in the context of a high-order SCAO system for METIS, one of the first-light instruments currently under design for the Extremely Large Telescope. We show that, though being iterative, the analyzed algorithms, when applied in the studied context, can be implemented in a very efficient manner, which reduces the related computational effort significantly. We demonstrate that the suggested analytically developed approaches involving iterative algorithms provide comparable quality to standard matrix-vector-multiplication methods while being computationally cheaper.
\end{abstract}

\section{Introduction} \label{sec_intro}

We consider the problem of fast, stable and highly accurate wavefront correction for large-scale real-time closed loop Adaptive Optics (AO) systems on Extremely Large Telescopes (ELTs). More specifically, we focus on reconstructing the wavefront $\Phi$ from pyramid wavefront sensor data $s=[s_x,s_y]$ with a pre-defined accuracy and within the required time slot to fit the real-time setting. The available measurements $s$ are related to the wavefront $\Phi$ by the non-linear integral operator $\boldsymbol P=[\boldsymbol P_x,\boldsymbol P_y]$ representing the pyramid sensor model. In part I of this paper \cite{HuSha18_1}, we extensively studied the physical and mathematical forward models of the pyramid wavefront sensor and the underlying operators. The major aim of this paper is to apply the theory of part I \cite{HuSha18_1} in order to solve the inverse problem of reconstruction the unknown wavefront from the given sensor data. Several approximations of the full Fourier optics based, non-linear pyramid sensor model derived in \cite{HuSha18_1} allow for a fast numerical implementation of the wavefront sensor (WFS) operators, which suggests to apply iterative algorithms for solving the reconstruction problem. \bigskip

We study and compare the performance (in terms of correction quality and speed) of well-known mathematical algorithms for solving inverse problems, namely conjugate gradient, steepest descent, Landweber, Landweber-Kaczmarz, and steepest descent-Kaczmarz iteration. The comparison is performed within the context of two specific high-order AO systems which are currently under design for the future ELT instruments. One of the systems is the SCAO (Single Conjugate Adaptive Optics) module for the METIS instrument~\cite{METIS_archiv} on a $39$~m telescope equipped with a $74 \times 74$~pyramid wavefront sensor (PWFS) sensing in the near-infrared K-band (at the wavelength $\lambda = 2200$~nm). The SCAO system is supposed to control $\sim 4000$ mirror actuators at frequencies of $500-1000$~Hz. The second system of interest is the eXtreme Adaptive Optics (XAO) module for the EPICS instrument~\cite{KoVe10} having a $200\times 200$~PWFS as a core component. For the XAO system a huge amount of $\sim 30000$ mirror actuators have to be controlled at a frequency of $3$~kHz, which is a challenge for the real-time control as the algorithm has to produce the reconstruction within a time of less than $0.33$~ms. Both considered systems are expected to provide perfect correction quality resulting in an unprecedented image contrast required for their scientific aims. \bigskip

Additionally, we compare the performance of our methods to the so far standard wavefront reconstruction algorithms used in these days. These are based on matrix-vector-multiplication (MVM) and invert the most exact Fourier optics model. However, the main drawback of any interaction-matrix-based method is their computational effort, which is demanding for the planned large-scale real-time AO systems.  The computational complexity required for setting up the command matrix scales as $O(N^3)$ with $N$ being the number of controlled actuators, and the application of this command matrix on the sensor data scales as $O(N^2)$ \cite{Ellerbroek02}. The indicated limitation of MVM methods makes their application on large-scale AO systems, such as the XAO system on the ELT, hardly feasible even on the hardware expected at the time of the telescope launch in around 2024. The steadily growing mirror sizes of future telescope systems imply an immense grow in the computational load of existing algorithms which means that the numerical effort has to be kept in mind when developing new algorithms. We demonstrate in this paper that all the proposed iterative algorithms provide the required high-quality and stable wavefront correction along with the heavily reduced computational complexity compared to standard matrix-vector-multiplication methods. The suggested matrix-free approaches make highly accurate real-time wavefront reconstruction feasible even for the XAO system.\bigskip

Since the pyramid wavefront sensor was first introduced in astronomical AO by Ragazzoni \cite{Raga96}, it has been extensively studied in optical test benches \cite{Esposito_05,Pinna_07,RaDi02,Veri04} and acknowledged to possess several precious characteristics distinguishing it from other types of wavefront sensors. Among those are the increased sensitivity, adjustable linearity range, and pupil sampling, as well as its ability to sense the segmented piston modes \cite{Esposito_05,Pinna_07}, which is gaining a special importance in the era of ELTs inevitably having segmented mirrors. Due to the named advantages, the PWFS is nowadays integrated as baseline on several telescope instruments under design. Among those we can name MICADO~\cite{Clenet_SPIE_2016_micado_pwfs}, HARMONI~\cite{Fusco_2010_spie_harmoni_eelt,Neichel_2016_spie_pwfs_harmoni_eelt}, METIS~\cite{METIS_archiv}, EPICS~\cite{KoVe10}, and ATLAS~\cite{Fusco_2010_spie_atlas_eelt} on the ELT, SCAO and LTAO systems on the Giant Magellan Telescope (GMT) \cite{VanDam_2012_pwfs_truth_LTAO_GMT,Esposito_2012_pwfs_NGS_SCAO_GMT}, as well as the NFIRAOS~\cite{Mieda_spie_2016_pwfs_truth_TMT,Veran_ao4elt4_pwfs_vs_sh} and PFI~\cite{Macintosh_2006_spie_pwfs_xao_tmt} instruments on the Thirty Meter Telescope (TMT). 
Furthermore, the application area of the PWFS is not limited to Adaptive Optics in Astronomy. Apart from astronomical observations, the pyramid sensor is also utilized in AO in ophthalmology~\cite{Chamot06,Daly_2010,Alvarez_Thesis,Iglesias02} and microscopy~\cite{Ig11,Ig13}. Hence, the problem of improving the quality of wavefront reconstruction approaches for this type of sensor by using sophisticated mathematical methods has never been more interesting, challenging, and important. \bigskip

In Section \ref{chap:existing_algorithms} we describe algorithms existing for wavefront reconstruction from pyramid sensor data. We sketch their advantages, drawbacks and limitations. We proceed with recalling the theoretical principles of wavefront sensing using the pyramid sensor in Section \ref{chap:theory_recall} and mention details on the discretization of the sensor in Section \ref{chap:discretization}. Afterwards, we describe several iterative algorithms, namely the conjugate gradient method for the normal equation and the steepest descent method, Landweber iteration as well as Kaczmarz type algorithms in Section \ref{chap:iter_algorithms}. This Section contains details on the numerical implementation of the involved operators as well. Section \ref{chap:numerics} presents the performance of the proposed algorithms and a comparison with respect to the achieved reconstruction quality. Finally, in Section \ref{chap:computational_complexity} we evaluate and compare the computational complexities of the analyzed approaches. Both, the reconstruction quality and the speed of the algorithms are additionally compared versus those of an MVM approach.

\section{Existing algorithms}\label{chap:existing_algorithms}
For the PWFS, several types of reconstruction approaches have been considered so far \cite{Hu18_thesis,ShatHut_spie2018_overview}. Frequently used are interaction-matrix-based methods which were already mentioned in Section \ref{sec_intro}.

A more computationally efficient algorithm, the Fourier Transform Reconstructor (FTR), with a complexity of $\mathcal{O}(N \log
N)$, was suggested in \cite{QuPa10}. The method was developed for a simplified geometrical optics based model of the roof
WFS, which assumes a large amount of modulation applied to the sensor. Under these assumptions the sensor data are modeled as the derivative of the phase, similar to the Shack-Hartmann (SH) sensor. In the reported algorithms, the authors applied SH Fourier domain filters for wavefront reconstruction. A slightly worse performance of the SH-based FTR compared to MVM was demonstrated for an $8$~m telescope. However, since the geometrical model is valid only for large modulations, this method requires further research on a modal optimization, which becomes especially important in the case of ELTs. 

The idea of applying inverse Fourier domain filters for wavefront reconstruction from pyramid sensor was further developed in \cite{Shat17_ao4elt5_clif,Shat17}, where the authors reported on two methods --- Convolution with the Linearized Inverse Filter (CLIF) with the complexity $\mathcal{O}(N^{3/2})$ and the Pyramid Fourier Transform Reconstructor (PFTR) with $\mathcal{O}(N\log N)$ complexity. Both algorithms use a precise Fourier optics forward model allowing the derivation of more exact Fourier filters connecting the incoming phase to the sensor measurements. 

The fastest available reconstruction algorithm is the Preprocessed Cumulative Reonstructor with Domain decomposition (P-CuReD) \cite{Shat_SPIE,Shat13}. Based on an analytical relation in the Fourier domain, pyramid data are transformed into SH-like data and subsequently inverted using CuReD \cite{Ros11,Ros12} which was originally developed for SH sensors and has already been tested on-sky \cite{Bitenc_cured_onsky,bitenc_ao4elt3_onsky}. This algorithm provides in numerical simulations the same or even better quality results than the standard interaction-matrix-based approaches and scales with a complexity of $\mathcal{O}(N)$.

Further algorithms are based on the fact that the wavefront reconstruction from the full pyramid sensor model can be simplified to an inversion of the finite Hilbert transform. The FHTR (Finite Hilbert Transform Reconstructor) algorithm proposed in \cite{Shatokhina_PhDThesis} uses the direct inversion formula of the finite Hilbert transform and a second method called SVTR (Singular Value Type Reconstructor) \cite{Hut17} is based on an analytical singular value expansion of the finite Hilbert transform operator. Both algorithms have the complexity $\mathcal{O}(N^{3/2})$.

\section{Pyramid/Roof wavefront sensors} \label{chap:theory_recall}
In the following, we recall the theory discussed in part I of the paper \cite{HuSha18_1}. Since all algorithms are based on an approximation of the full pyramid sensor, or more precisely on a linearization of the roof sensor, we will only focus on the properties of the latter and mention corresponding adjoint operators which are needed for the application of the proposed iterative methods. For a precise analysis of the full pyramid sensor model, details about the linearization procedure and proofs we refer the reader to part I of the paper \cite{HuSha18_1}.  

The analytical models of the non- and modulated pyramid wavefront sensors are complex and therefore difficult to invert directly. However, some assumptions suggested by the physical setting of the model itself, allow to simplify the non-linear Fourier optics based model of the pyramid sensor $\boldsymbol P = [\boldsymbol P_x,\boldsymbol P_y]$ (cf \cite{HuSha18_1} for the definition) by substituting the pyramidal prism by two orthogonally placed two-sided roof prisms \cite{BuDa06, Phill06,Veri04}. Due to the physical decoupling of the roof prisms and their orthogonality, the two signal sets $s=[s_x,s_y]$ provided by the pyramid sensor as
\begin{equation}\label{meas1}
\begin{split}
s_x &= -\tfrac{1}{2}\boldsymbol P_x\Phi \\
s_y &=\ \ \tfrac{1}{2}\boldsymbol P_y\Phi
\end{split}
 \end{equation}
 become independent and contain information about the incoming phase $\Phi$ only in $x$- and $y$-direction respectively. Because of symmetry, we only consider the roof sensor operator $\boldsymbol R=\left[\boldsymbol R_x,\boldsymbol R_y\right]$ in $x$-direction, i.e., $\boldsymbol R_x$. By interchanging $x$ and $y$ all assertions are derived for $\boldsymbol R_y$ accordingly.    
 
We describe the annular telescope aperture mask by $\Omega=\Omega_y\times \Omega_x \subseteq \left[-D/2,D/2\right]^2$. Single lines (intervals) of the annular aperture are represented by $\Omega_x = \left[a_x,b_x\right]$ and $\Omega_y = \left[a_y,b_y\right]$, with $a_x <b_x$, $a_y<b_y$ being the borders of the pupil for fixed $x$  and $y$ correspondingly.

\subsection{Linearized roof WFS}
\begin{theorem}
Under the roof sensor assumption, the PWFS signal corresponding to no, circular and linear modulation of amplitude $\alpha = \tfrac{b\lambda}{D}$ with a positive integer $b$ is approximated by
\begin{equation*}\label{roof_meas}
s_x^{\{n,c,l\},lin} =-\tfrac{1}{2} \left(\boldsymbol R_x^{\{n,c,l\},lin}\Phi\right)(x,y)
\end{equation*}
with 
\begin{equation}\label{roof_operator}
\begin{split}
\left(\boldsymbol R_x^{\{n,c,l\},lin}\Phi\right)(x,y)&:=\mathcal{X}_{\Omega}(x,y)\dfrac{1}{\pi} \int\limits_{\Omega_y}{\dfrac{ [\Phi(x',y)-\Phi(x,y)] \cdot k^{\{n,c,l\}} (x'-x) }{x'-x}\ dx'}, \\
\left(\boldsymbol R_y^{\{n,c,l\},lin}\Phi\right)(x,y) &:= \mathcal{X}_{\Omega}(x,y)\dfrac{1}{\pi} \int\limits_{\Omega_x}{\dfrac{{\left[\Phi(x,y')-\Phi(x,y)\right]}k^{\{n,c,l\}}(y'-y)}{y'-y}\ dy'}
\end{split}
\end{equation}
indicating the linearized roof sensor operators and
\begin{align*}
k^n (x) &:= 1   \qquad \qquad \qquad \qquad \qquad \text{(no modulation)}, \\
k^c (x)&: = J_0 (\alpha_{\lambda} x) \qquad \qquad \qquad \ \text{(circular modulation)}, \\
k^l(x)&:=\sinc(\alpha_{\lambda} x) \qquad \qquad \quad \ \ \text{(linear modulation)}. 
\end{align*}
The modulation parameter is given by $\alpha_\lambda = \tfrac{2\pi\alpha}{\lambda}$ and $J_0$ denotes the zero-order Bessel function of the first kind.
\end{theorem}
\begin{proof}
See \cite{BuDa06,HuSha18_1,Veri04}.
\end{proof}
The linearized roof sensor operators $\boldsymbol R^{\{n,c,l\},lin}$ offer a further possibility for simplification of the model due to the splitting
\begin{equation}\label{eq:3c.133}
\left( \boldsymbol R^{\{n,c,l\},lin}_x\Phi \right) (x,y) = \mathcal{X}_{\Omega}(x,y)\left[\left( \boldsymbol L_x^{\{n,c,l\}} \Phi \right) (x,y) - \Phi (x,y) \cdot \left(\boldsymbol L_x^{\{n,c,l\}} 1\right) (x,y)\right]
\end{equation}
for the integral operators $\boldsymbol L_x^{\{n,c,l\}}: \mathcal{H}^{11/6}\left(\R^2\right)\rightarrow\mathcal{L}_2\left(\R^2\right)$ defined by
\begin{equation}\label{3c.13}
( \boldsymbol L_x^{\{n,c,l\}} \Phi ) (x,y) := \dfrac { 1 } { \pi } \ p.v. \int \limits_{\Omega_y} \dfrac{  \Phi (x',y) k^{\{n,c,l\}} \left(x'-x\right)} {x'-x} \ dx',
\end{equation}
where $p.v.$ denotes the Cauchy principal value. Note that $\boldsymbol L_x^n$ is the finite Hilbert transform operator. Dropping the second term in \eqref{eq:3c.133} leads to the inverse problem 
\begin{equation*} \label{meas_Lx}
s_x =-\tfrac{1}{2} \boldsymbol L_x^{\{n,c,l\}}\Phi
\end{equation*}
 for pyramid sensor data $s_x$.

\subsection{Adjoint operators}
All iterative methods for wavefront reconstruction from pyramid sensor data proposed below require the application of adjoint operators. As mentioned, e.g., in \cite{HuSha18_1}, it holds that $\Phi \in \mathcal{H}^{11/6}$, i.e., the underlying operators are defined as $\boldsymbol L^{\{n,c,l\},lin}_x: \mathcal{H}^{11/6} \rightarrow \mathcal{L}_2$ and  $\boldsymbol R^{\{n,c,l\},lin}_x: \mathcal{H}^{11/6} \rightarrow \mathcal{L}_2$. Together with the embedding operator $i_s: \mathcal{H}^{11/6} \rightarrow \mathcal{L}_2$ we derive the corresponding adjoint operators by $$\left( \boldsymbol L^{\{n,c,l\}}_x\right)^* = i_s^*\left( \tilde{\boldsymbol L}^{\{n,c,l\}}_x\right)^*$$ with $\left( \tilde{\boldsymbol L}^{\{n,c,l\}}_x\right)^*:\mathcal{L}_2\rightarrow \mathcal{L}_2$ according to \cite{RaTesch04}. Hence, it sufficies to calculate $\left( \tilde{\boldsymbol L}^{\{n,c,l\}}_x\right)^*$. For simplicity, we use the notation $\left(\tilde{\boldsymbol L}^{\{n,c,l\}}_x\right)^*$ instead of $\left( \boldsymbol L^{\{n,c,l\}}_x\right)^*$ in the following. The roof sensor operators are considered accordingly.

\begin{prop}
The adjoint operators of the roof sensor and its one-term approximation in $\mathcal{L}_2\left(\R^2\right)$ are given by
\begin{align*}
\left( \left( \boldsymbol L^{\{n,c,l\}}_x\right)^*\Psi\right)(x,y) &= -\dfrac{1}{\pi}\ p.v.\int\limits_{\Omega_y}{\dfrac{\Psi(x',y) \cdot k^{\{n,c,l\}} (x'-x) }{x'-x}\ dx'}, \\
\left( \left( \boldsymbol R^{\{n,c,l\},lin}_x\right)^*\Psi\right)(x,y) &= -\dfrac{1}{\pi}\ p.v.\int\limits_{\Omega_y}{\dfrac{\left[\Psi(x',y)+\Psi(x,y)\right] \cdot k^{\{n,c,l\}} (x'-x) }{x'-x}\ dx'}.
\end{align*}
\end{prop}
\begin{proof}
See \cite{HuSha18_1}.
\end{proof}

\section{The discrete pyramid wavefront sensor}\label{chap:discretization}

The full continuous measurements $s_x(x,y)$ and $s_y(x,y)$ of the pyramid wavefront sensor are not available in practice. For the description of the discrete pyramid sensor we perform a division of the continuous two dimensional process into finitely many equispaced regions called subapertures. The data are then assumed to be averaged over every subaperture which corresponds to the finite sampling of the pyramid sensor. Note that in reality, the subaperture grid is predefined by the sensor's physics. Following the approach in \cite{Veri04}, we examine the sensor data as functions evaluated in the (discrete) middle points of the WFS subapertures. In the two dimensional case we consider quadratic subapertures of size $d \times d$ with $d=\frac{D}{n}$, where $D$ represents the telescope diameter, i.e., the primary mirror size, and $n$ the number of subapertures in one direction. \bigskip \par
Note that all considerations are valid for measurements both in $x$-direction $s_x(x,y)$ and $y$-direction $s_y(x,y)$, as well as for non-modulated, circularly, and linearly modulated data. Thus, we consider general measurements identified by $s(x,y)$. Discretizing $s(x,y)$ delivers $n^2$ data values $s_{jk}$ with $j,k=1,\dots , n$. \bigskip \par

For the following, we use the Dirac comb $\dirac_d$ defined as
\be{2.1111} \notag
\dirac_d\left(x,y\right) := \sum_{\ell=-\infty}^\infty{\sum_{m=-\infty}^\infty{\delta\left(x-\ell d,y-m d\right)}}.
\ee
where $\delta$ denotes the delta distribution. \bigskip\par 
The continuous signal is captured by the wavefront sensor as follows: \par
First, the average of the measurements over one subaperture is calculated. This is represented as a convolution of the continuous data $s(x,y)$ with a characteristic function $ \mathcal{X}_{\left[-d/2,d/2\right]^2}\left(x,y\right)$, i.e., 
\begin{align*}\label{eq:2.1}
\tilde{s}(x,y) 
&=\frac{1}{d^2}\int_{x-d/2}^{x+d/2}{\int_{y-d/2}^{y+d/2}{s(x',y') \ dy'} \ dx'} \\ \notag
&= \frac{1}{d^2}\int_{-\infty}^{\infty}{\int_{-\infty}^\infty{s(x',y')\cdot\mathcal{X}_{\left[-d/2,d/2\right]}\left(x-x'\right)\cdot\mathcal{X}_{\left[-d/2,d/2\right]}\left(y-y'\right)dy'} \ dx'} \\ \notag
&= \frac{1}{d^2}\int_{-\infty}^{\infty}{\int_{-\infty}^\infty{s(x',y')\cdot\mathcal{X}_{\left[-d/2,d/2\right]^2}\left(x-x',y-y'\right)dy'} \ dx'} \\ \notag
&= \frac{1}{d^2}\left(s \ast \mathcal{X}_{\left[-d/2,d/2\right]^2}\right)\left(x,y\right). \notag
\end{align*}
The discretization is carried out as a application of the Dirac comb $\dirac_d$ assuming that the measurements $s$ fulfill the necessary conditions on applying the distribution $\delta$. Herewith, we assign a discrete set of measurements $\overline{s}$ centered on the subapertures
\begin{equation*}
\overline{s}=  \langle  \dirac_d,\tilde{s}\rangle =\langle \sum_{\ell=-\infty}^{\infty}{\sum_{m=-\infty}^\infty{\delta\left(\cdot-\ell d,\cdot-m d\right)}},\tilde{s}\rangle
\end{equation*}
from floating average values $\tilde{s}(x,y)$ to the discrete set of subaperture middle points $\{\left(jd,kd\right):j,k \in \Z \}$. 

Finally, we restrict the number of measurements to the size of the region captured by the sensor. For several telescope systems the pupil $\Omega$ is annular instead of circular since a shade created by the secondary mirror prevents measurements on these areas. Thus, the light in the area of the central obstruction possibly does not produce reliable measurements~\cite{engler_spie_2018,HuShaOb18,ObRafShaHu18_proc,schwartz_spie2018,schwartz_ao4elt5}. 

Hence, the last step is realized as a multiplication with a second characteristic function
 \be{2.3} \notag
s_{j,k} = \left(\langle \dirac_d,\tilde{s}\rangle \cdot \mathcal{X}_{\Omega}\right)_{j,k}
\ee
for $j,k = 1,\dots, n$. To be precise, we usually consider less than $n^2$ measurements due to the annular shape of the aperture and ignore subapertures which are too less illuminated in order to produce reliable data. However, we will not specifically mention this fact throughout the paper.

\section{Iterative wavefront reconstruction methods}  \label{chap:iter_algorithms}
In this Section, we adapt well-known mathematical algorithms, namely the conjugate gradient method for the normal equation, the steepest descent algorithm, and Landweber iteration as well as modifications of these methods coupled with a Kaczmarz strategy to the problem of wavefront reconstruction from pyramid sensor data. By solving the WFS equations in the operator setting instead of forming matrices we minimize the computational complexity of the proposed methods.  
For wavefront reconstruction we solve the two integral equations  
\begin{align}\label{5.1a}
-\tfrac{1}{2} \  \boldsymbol R_x^{\{n,c,l\},lin}\Phi &=s_x \\ \label{5.1b}
\tfrac{1}{2} \  \boldsymbol R_y^{\{n,c,l\},lin}\Phi &=s_y
\end{align}
with $ \left[\boldsymbol R_x^{\{n,c,l\},lin},\boldsymbol R_y^{\{n,c,l\},lin}\right]$ representing a linearization of the roof WFS according to \eqref{roof_operator} and  $s=\left[s_x,s_y\right]$ pyramid sensor measurements. As a further simplification, we consider the inverse problem
\begin{equation}\label{5.1c}
\begin{split}
-\tfrac{1}{2} \ \boldsymbol L_x^{\{n,c,l\}}\Phi &=s_x \\
\tfrac{1}{2} \ \boldsymbol L_y^{\{n,c,l\}}\Phi &=s_y.
\end{split}
\end{equation}
For simplicity of notation we use $\boldsymbol Q := \tfrac{1}{2}\cdot  \left[-\boldsymbol R_x^{\{n,c,l\},lin},\boldsymbol R_y^{\{n,c,l\},lin}\right]$ in the following since the basic idea is the same for all types of modulation. Moreover, we only concentrate on solving the inverse problem \eqref{5.1a}-\eqref{5.1b}, but mention that solutions of \eqref{meas1} and \eqref{5.1c} can be calculated accordingly.\bigskip

To specify the representation of the incoming phase $\Phi$ and the measurements $s$ we denote the number of subapertures by $n$. There are various possible representations for the phase and the measurements, e.g., Zernike polynomials or bilinear spline functions.  We choose a representation that guarantees maximum computational efficiency, and thus assume that the incoming phase and the measurements are piecewise constant on the subapertures, i.e., 
\be{5.5} \notag
\Phi(x,y) = \sum\limits_{i=1}^n{\phi_i\mathcal{X}_{\Omega_i^y}(x)}, \qquad \qquad \qquad \qquad  s_x(x,y) = \sum\limits_{i=1}^n{s_{x,i}\mathcal{X}_{\Omega_i^y}(x)},
\ee
where $\left(\phi_i\right)_{1\le i \le n}, \left(s_{x,i}\right)_{1\le i \le n}$ denote basis coefficients and $\Omega_i^y = \left[x_{i-1}^y,x_i^y\right]$ the $i$-th subaperture of a row for fixed $y$. As the wavefront sensor provides two measurements (one in $x$- and one in $y$-direction), for every single subaperture, the suggestion of representing the measurements via piecewise constant functions describing the subaperture grid is reasonable. 
We calculate the involved operators as, e.g.,
{\small\begin{align}\label{5.6}
\left(\boldsymbol Q_x\Phi\right)(x,y) &=- \mathcal{X}_{\Omega}(x,y)\dfrac{1}{2\pi}  \int\limits_{\Omega_y}{\dfrac{\left[\Phi(x',y)-\Phi(x,y)\right]k^{\{n,c,l\}}(x'-x)}{x'-x}\ dx'} \notag \\ \notag
&=-\mathcal{X}_{\Omega}(x,y)\dfrac{1}{2\pi}  \int\limits_{\Omega_y}{\dfrac{\left[\sum\limits_{i=1}^n{\phi_i\mathcal{X}_{\Omega_i^y}(x')}-\sum\limits_{i=1}^n{\phi_i\mathcal{X}_{\Omega_i^y}(x)}\right]k^{\{n,c,l\}}(x'-x)}{x'-x}\ dx'} \\ \notag
&= -\mathcal{X}_{\Omega}(x,y)\dfrac{1}{2\pi} \ p.v. \left[\int\limits_{\Omega_y}{\dfrac{\sum\limits_{i=1}^n{\phi_i\mathcal{X}_{\Omega_i^y}(x')}k^{\{n,c,l\}}(x'-x)}{x'-x}\ dx'} - \int\limits_{\Omega_y}{\dfrac{\sum\limits_{i=1}^n{\phi_i\mathcal{X}_{\Omega_i^y}(x)}k^{\{n,c,l\}}(x'-x)}{x'-x}\ dx'} \right] \\ \notag
&=  -\mathcal{X}_{\Omega}(x,y)\dfrac{1}{2\pi}\sum\limits_{i=1}^n\phi_i \ \underbrace{ p.v. \  \left[  \ \int_{x_{i-1}^y}^{x_i^y}{\dfrac{k^{\{n,c,l\}}(x'-x)}{x'-x}\ dx' }-\mathcal{X}_{\Omega_i^y}(x) \int\limits_{\Omega_y}{\dfrac{k^{\{n,c,l\}}(x'-x)}{x'-x}\ dx'}\right]}_{=:\ \alpha_i^{\{n,c,l\}}\left(x,y\right)} \\
&=-\mathcal{X}_{\Omega}(x,y)\dfrac{1}{2\pi}\sum\limits_{i=1}^n{\phi_i\alpha_i^{\{n,c,l\}}(x,y)}.
\end{align}
}The functions $\alpha_i^{\{n,c,l\}}(x,y)$ are computed offline and do not influence the computational speed of the proposed methods. The implementation of all involved operators is performed analogously when choosing the basis representation \eqref{eq:5.5}. \bigskip

Now, we focus on concrete wavefront reconstruction algorithms for pyramid sensor data using iterative methods. In particular, we consider the conjugate gradient algorithm for the normal equation (CGNE), the steepest descent (SD), and Landweber iteration. 
Since the pyramid sensor provides two measurements $s_x$ and $s_y$, the above named approaches deliver two solutions $\Phi = \left[\Phi_x,\Phi_y\right]$, one in $x$- and one in $y$- direction. The final reconstruction $\Phi^{rec}$ is then computed as the average of the two temporary solutions. An alternative combination of the two measurements $s_x$ and $s_y$ using Kaczmarz loops is investigated in Section \ref{sec:Kaczmarz}. Note that the theory of the presented algorithms is mainly based on \cite{Engl,Louis89}. The considered norms are the $\mathcal{L}_2$-norms. 


\subsection{CGNE approach}
The conjugate gradient (CG) method is one of the most powerful algorithms for solving self-adjoint, positive (semi-)definite linear equations \cite{Bra87,Engl,Gil77,Hanke95,Hes52,KaNa72}. For solving the wavefront reconstruction problem we apply the conjugate gradient method to the normal equation
\begin{equation}\label{norm_eq}
\boldsymbol Q^*\boldsymbol Q\Phi = \boldsymbol Q^*s.
\end{equation}
Let $\boldsymbol Q^\dagger$ denote the Moore-Penrose generalized inverse. The CG-iterates $\left(\Phi_i\right)$ converge to $\boldsymbol Q^\dagger s$ for all $s \in \mathcal{D}(\boldsymbol Q^\dagger)$ \cite{Engl} by requiring the fewest iterations among all semiiterative methods.  \bigskip

The CGNE method (Algorithm $1$) applied to the inverse problem of wavefront reconstruction from pyramid wavefront sensor data is described by: \bigskip\par
\begin{tabular}{l c r}
\hline
\textbf{Algorithm 1} CGNE for pyramid sensors \\
\hline
choose $\Phi_0$, initialize $d_0=s-\boldsymbol Q\Phi_0$, $p_1=s_0=\boldsymbol Q^*d_0$ \\
for $i=1,\dots K$ do \\
\quad $q_i = \boldsymbol Q p_i$ \\
\quad $ \alpha_i = ||s_{i-1}||^2/||q_i||^2$\\ 
\quad $\Phi_i = \Phi_{i-1}+\alpha_ip_i$\\
\quad $d_i = d_{i-1}-\alpha_iq_i$ \\
\quad $s_i = \boldsymbol Q^*d_i$ \\
\quad $\beta_i = ||s_i||^2/||s_{i-1}||^2$ \\
\quad $p_{i+1} = s_i+\beta_ip_i$ \\
endfor \\
$\Phi^{rec} =  \left(\Phi_{x,K}+\Phi_{y,K}\right)/2$ \\
\hline
\end{tabular}
\vspace*{1cm}

\subsection{Steepest descent approach}

For solving the system \eqref{5.1a}-\eqref{5.1b} we are additionally interested in the method of steepest descent where we consider different choices of the step sizes in the iterative process. For pyramid sensors, we use the SD method (Algorithm $2$) applied to the least-squares functional
\be{5.8}
J\left(\Phi\right) = \left|\left|\boldsymbol Q\Phi-s\right|\right|^2_{\mathcal{L}_2} \rightarrow \min.
\ee
The method of steepest descent was originally introduced by Cauchy \cite{Cauchy1847} as one of the most basic procedures to minimize a differentiable functional. A popular step size is determined by an exact line search in the direction of the negative gradient. Alternative choices of the step size have already been considered, e.g., in \cite{Sa16,SaRa15} and will be discussed below for the problem of wavefront reconstruction from pyramid data using the SD method. The gradient of the classical least-squares functional is given by
\be{5.9}
J'\left(\Phi\right) = \boldsymbol Q^*\left(\boldsymbol Q\Phi-s\right)
\ee
and the resulting algorithm reads as: \bigskip\par
\begin{tabular}{l c r}
\hline
\textbf{Algorithm 2} Steepest descent method for pyramid sensors \\
\hline
choose $\Phi_0$ \\
for $i=1,\dots K$ do \\
\quad $d_{i-1} = -J'\left(\Phi_{i-1}\right)$ \\
\quad $ \tau_{i-1}= \min\limits_{t \in [0,\infty)} J\left(\Phi_{i-1}+td_{i-1}\right)$\\ 
\quad $ \Phi_i = \Phi_{i-1} + \tau_{i-1}d_{i-1}$\\ 
endfor \\
$\Phi^{rec} =  \left(\Phi_{x,K}+\Phi_{y,K}\right)/2$ \\
\hline
\end{tabular}
\vspace*{1cm}


\subsubsection{Step size choices and convergence}

The speed of convergence of the gradient iteration
\begin{align*}
\Phi_{i}& = \Phi_{i-1} + \tau_{i-1}d_{i-1} \\
d_{i-1}&=-J'(\Phi_{i-1})
\end{align*}
depends highly on the choice of the step size $\tau_{i}$. We consider the classical steepest descent (line search) step size that is defined by $$ \tau_i^{SD} = \min\limits_{t\in [0.\infty)} J\left(\Phi_i+td_i\right).$$ This means that an exact line search is performed in the direction of steepest descent which corresponds to the direction of the negative gradient. For the least-squares functional \eqref{eq:5.8} and corresponding derivative \eqref{eq:5.9} the steepest descent step size with $d_i = -\boldsymbol Q^*\left(\boldsymbol Q\Phi-s\right)$ reads as 
\begin{equation}\label{tau_SD}
\tau_i^{SD} = \dfrac{\left|\left|d_i\right|\right|^2}{\left|\left|\boldsymbol Qd_i\right|\right|^2}
\end{equation}
and results in the so called \textit{Cauchy method}.

If we minimize the gradient norm along the search direction, we obtain another line search method for finite dimensions, namely the method of \textit{minimal gradient} (MG) \cite{Dai03} given by 
\begin{equation}\label{tau_MG}
 \tau_i^{MG} = \dfrac{\left|\left|\boldsymbol Qd_i\right|\right|^2}{\left|\left|\boldsymbol Q^*\boldsymbol Qd_i\right|\right|^2}. 
\end{equation} 
From Cauchy-Schwarz inequality it follows $\tau_i^{MG}\le \tau_i^{SD}$.

Because of  zigzagging between consecutive steps the SD method suffers from slow convergence in some cases. To overcome these effects, a fast and efficient alternative step size choice was introduced by \textit{Barzilai and Borwein} (BB) in \cite{Bar88}. The BB technique is motivated by quasi-Newton methods and derived from a two-point approximation to the secant equation. There exist two versions of the BB method which are defined by 
$$ \tau_i^{BB1} = \dfrac{\langle \Delta \Phi_i,\Delta d_i\rangle}{\langle \Delta d_i,\Delta d_i\rangle} \qquad \qquad \text{and} \qquad \qquad  \tau_i^{BB2} = \dfrac{\langle \Delta \Phi_i,\Delta \Phi_i\rangle}{\langle \Delta \Phi_i,\Delta d_i\rangle}$$
with $\Delta \Phi_i = \Phi_i-\Phi_{i-1}$ and $\Delta d_i = d_i-d_{i-1}$. Plugging in the calculations corresponding to Algorithm $2$ we obtain
\begin{align*}
\Delta \Phi_i &=  \Phi_i-\Phi_{i-1} = \tau_{i-1}d_{i-1} \\
\Delta d_i &= d_i-d_{i-1} = -\boldsymbol Q^*\boldsymbol Q\left(\Phi_i-\Phi_{i-1}\right) = -\tau_{i-1}\boldsymbol Q^*\boldsymbol Qd_{i-1} \\
d_i &=\left(I-\tau_{i-1}\boldsymbol Q^*\boldsymbol Q\right)d_{i-1}
\end{align*}
since the involved operators are linear. Therefore, the BB step sizes are rewritten as
$$\tau_i^{BB1}=\dfrac{\left|\left|\boldsymbol Qd_{i-1}\right|\right|^2}{\left|\left|\boldsymbol Q^*\boldsymbol Qd_{i-1}\right|\right|^2} \qquad \qquad \text{and} \qquad \qquad \tau_i^{BB2}=\dfrac{\left|\left|d_{i-1}\right|\right|^2}{\left|\left|\boldsymbol Q d_{i-1}\right|\right|^2},$$ 
i.e., $$\tau_i^{BB1} = \tau_{i-1}^{MG} \qquad \qquad \text{and} \qquad \qquad \tau_i^{BB2}=\tau_{i-1}^{SD}.$$ The idea is to use additional information of the previous iteration to compute the step size for the current iteration. Once again with the Cauchy-Schwarz inequality we obtain $\tau_i^{BB1}\le\tau_i^{BB2}$. Generally, while the SD and MG method decrease monotonically, the BB step size choices are non-monotone as the error behaves non-monotonously, i.e., $\left|\left|\Phi-\Phi_{i+1}\right|\right|\le\left|\left|\Phi-\Phi_i\right|\right|$ for the true solution $\Phi$ is not fulfilled for every iteration $i$. Nevertheless, the BB method converges to a solution of \eqref{eq:5.8} as found in \cite{Ray93}.\bigskip

The \textit{Cauchy-Barzilai-Borwein} (CBB) step size is based on the idea to use the SD and BB step size alternating. The method, which was introduced in 2003 and is also called \textit{alternate step size} (AS) \textit{gradient method}, aims at reducing the zigzag-effect of the Cauchy method \cite{Dai03_1}, and therefore leads to a faster convergence. The promising alternative to the BB method reads as 
$$ \tau_i^{CBB} = \tau_i^{AS} = \begin{cases} \tau_i^{SD}, \qquad \  \text{for i odd}, \\ \tau_i^{BB2}, \qquad \text{for i even}. \end{cases} $$
Due to $\tau_i^{BB2} = \tau_{i-1}^{SD}$, we use the same step size twice in two consecutive iterations. 
An alternate version of the MG method called \textit{alternate minimization} (AM) \textit{gradient method} was proposed in \cite{Dai03} having an SD iteration for every second step. Generally, the SD method becomes faster when one non-monotone (e.g., BB) step is made even after several SD steps \cite{Zhou06}.
In addition, a variety of step size choices have been introduced using combinations or shortened step size versions of the above mentioned options. \bigskip

To reduce the computational effort of the algorithms, one can use a fixed step size. Step sizes which are tuned heuristically depend mainly on the size of the telescope and the resolution of the wavefront sensor (discretization). For a fixed step size $\tau_i = \beta$ the steepest descent algorithm for the least-squares functional \eqref{eq:5.8} reduces to the standard Landweber iteration.


\subsection{Landweber approach}

For the Landweber iteration \cite{Landweber51}, the normal equation \eqref{norm_eq} is transformed into the equivalent fixed point equation
$$\Phi = \Phi+\boldsymbol Q^*\left(s-\boldsymbol Q\Phi\right).$$
In order to ensure convergence by $\left|\left|\boldsymbol Q\right|\right|\le 1$ we introduce a relaxation parameter $0<\beta\le \left|\left|\boldsymbol Q\right|\right|^{-2}$ and iterate by
$$ \Phi_{i} = \Phi_{i-1} + \beta \boldsymbol Q^*\left(s-\boldsymbol Q\Phi_{i-1}\right), \qquad i \in \N.$$
Then, $\left(\Phi_i\right)$ converges to a least-squares solution of \eqref{5.1a}-\eqref{5.1b} for $s \in \mathcal{D}(\boldsymbol Q^\dagger)$ \cite{Engl}. \bigskip
 
The Landweber iteration modified for wavefront reconstruction based on pyramid sensor measurements (Algorithm $3$) reads as: \bigskip\par
 
\begin{tabular}{l c r}
\hline
\textbf{Algorithm 3} Landweber iteration for pyramid sensors \\
\hline
choose $\Phi_0$, set relaxation parameter $\beta$ \\
for $i=1,\dots K$ do \\
\quad $ \Phi_{i} = \Phi_{i-1} + \beta \boldsymbol Q^*\left(s-\boldsymbol Q\Phi_{i-1}\right)$\\ 
endfor \\
$\Phi^{rec} =  \left(\Phi_{x,K}+\Phi_{y,K}\right)/2$ \\
\hline
\end{tabular}
\vspace*{1cm}

Besides the above discussed methods for pyramid sensors, there already exist several algorithms providing two reconstructions, one from data $s_x$ and one from data $s_y$ \cite{Hut17,Shatokhina_PhDThesis,Shat17}. Since in the reconstructions obtained by averaging the two solutions we experienced distinct horizontal and vertical artifacts, the aim is to combine the reconstructions already during the iteration steps. For this reason we investigate Kaczmarz methods in which the two data sets are used alternating.


\subsection{Kaczmarz methods for wavefront reconstruction from pyramid sensor data} \label{sec:Kaczmarz}

For the reconstruction of the incoming wavefront, the pyramid sensor provides two data sets $s_x$ and $s_y$. If the reconstruction were based on the full pyramid model, the incoming phase $\Phi$ either could be reconstructed solely from measurements $s_x$ or solely from $s_y$ because the null space of the operators consists only of the global piston mode, which anyway does not influence the imaging quality. However, in case the reconstruction algorithms utilize the roof sensor model, both data need to be used due to different null spaces of the single operators. Altogether, there are several facts that support the usage of both data sets. On the one hand, we expect better reconstruction quality in case we utilize more information. This argument is additionally strengthened by the presence of noise in the sensor measuring process. On the other hand, deeper investigations of the underlying operators in $x$- and $y$-direction show that they have different null spaces, i.e., depending on the underlying model of the reconstructors there exist modes that cannot be reconstructed. For instance, pyramid and roof wavefront sensors are not able to detect a constant added to the incoming phase $\Phi$. This undetectable constant, called piston mode (mode of order $0$), has no influence on the measurements $s$. 
In order to characterize effects that are invisible in sensor data we discuss selected wavefront modes (of order $0$ and $1$) which are elements of the null space of  the roof wavefront sensor operators, i.e., phase elements that deliver measurements equal to zero. For the following investigations, we will consider the mathematical forward model of the linearized roof sensor $\boldsymbol R^{lin}$ and analyze the null spaces of the corresponding operators described by
$$\mathcal{N}\left(\boldsymbol R^{lin}\right):=\{\Phi\in \mathcal{H}^{11/6}\left(\R^2\right)| \ \boldsymbol R^{lin}\Phi = 0\}.$$

We study the response of the linearized roof sensor to a global piston mode shown in Figure \ref{fig:null_space} left. Hence, we define $$\Phi_{piston}(x,y)=c\cdot \mathcal{X}_{\Omega}(x,y),$$ where $c \in \R$ is a constant. Furthermore, we analyze how the sensor responses to modes of order $1$ called tip \& tilt modes (see Figure \ref{fig:null_space} middle and right) represented by $$\Phi_{tip/tilt}(x,y)=\left(ax+by\right)\cdot \mathcal{X}_{\Omega}(x,y)$$ for $a,b \in \R$.

\begin{prop}\label{4.2}
Constant functions $c\cdot \mathcal{X}_{\Omega}$ with $c \in \mathbb{R}$ are elements of the null space of the linearized roof sensor operators $\boldsymbol R^{\{n,c,l\},lin} = \left[\boldsymbol R_x^{\{n,c,l\},lin}, \boldsymbol R_y^{\{n,c,l\},lin}\right]$ with $\boldsymbol R_x^{\{n,c,l\},lin}$ defined in \eqref{roof_operator}. Moreover, tip signals $cx\cdot \mathcal{X}_{\Omega}(x,y)$ are in the null space of $\boldsymbol R_y^{\{n,c,l\},lin}$ and tilt signals $cy\cdot \mathcal{X}_{\Omega}(x,y)$ are in the null space of $\boldsymbol R_x^{\{n,c,l\},lin}$.
\end{prop}

\begin{proof}
Global phase piston modes $c\cdot \mathcal{X}_{\Omega}$ are in the null space of the roof sensor operators because of
\begin{align*}
\left(\boldsymbol R_x^{\{n,c,l\},lin}\Phi_{piston}\right)(x,y) &= \mathcal{X}_{\Omega}(x,y)\dfrac{1}{\pi} \int\limits_{\Omega_y}{\dfrac{\left[\Phi_{piston}(x',y)-\Phi_{piston}(x,y)\right]k^{\{n,c,l\}}(x'-x)}{x'-x}\ dx'} \\
& = \mathcal{X}_{\Omega}(x,y)\dfrac{1}{\pi} \int\limits_{\Omega_y}{\dfrac{\left[c-c\right]k^{\{n,c,l\}}(x'-x)}{x'-x}\ dx'} =0
\end{align*}
and $\boldsymbol R_y^{\{n,c,l\},lin}(x,y)$ respectively. \bigskip

For exact investigations of tip \& tilt modes, we split $\Phi_{tip/tilt}(x,y)$ into $$\Phi_{tip}(x,y) = ax\cdot  \mathcal{X}_{\Omega}(x,y) \qquad \text{and} \qquad  \Phi_{tilt}(x,y) = by\cdot \mathcal{X}_{\Omega}(x,y).$$
Then, we consider
\begin{align*}
\left(\boldsymbol R_x^{\{n,c,l\},lin}\Phi_{tip}\right)(x,y) &= \mathcal{X}_{\Omega}(x,y)\dfrac{1}{\pi} \int\limits_{\Omega_y}{\dfrac{{\left[\Phi_{tip}(x',y)-\Phi_{tip}(x,y)\right]}k^{\{n,c,l\}}(x'-x)}{x'-x}\ dx'} \\
&= \mathcal{X}_{\Omega}(x,y)\dfrac{1}{\pi} \int\limits_{\Omega_y}{\dfrac{{\left[a\left(x'-x\right)\right]}k^{\{n,c,l\}}(x'-x)}{x'-x}\ dx'} \\
&= \mathcal{X}_{\Omega}(x,y)\dfrac{1}{\pi} \int\limits_{\Omega_y}{a\cdot k^{\{n,c,l\}}(x'-x)\ dx'} 
\end{align*}
and
\begin{align*}
\left(\boldsymbol R_y^{\{n,c,l\},lin}\Phi_{tip}\right)(x,y) &= \mathcal{X}_{\Omega}(x,y)\dfrac{1}{\pi} \int\limits_{\Omega_x}{\dfrac{{\left[\Phi_{tip}(x,y')-\Phi_{tip}(x,y)\right]}k^{\{n,c,l\}}(y'-y)}{y'-y}\ dy'} \\
&=0
\end{align*}
as well as
\begin{align*}
\left(\boldsymbol R_x^{\{n,c,l\},lin}\Phi_{tilt}\right)(x,y) &= \mathcal{X}_{\Omega}(x,y)\dfrac{1}{\pi} \int\limits_{\Omega_y}{\dfrac{{\left[\Phi_{tilt}(x',y)-\Phi_{tilt}(x,y)\right]}k^{\{n,c,l\}}(x'-x)}{x'-x}\ dx'} \\
&= 0
\end{align*}
and
\begin{align*}
\left(\boldsymbol R_y^{\{n,c,l\},lin}\Phi_{tilt}\right)(x,y) &= \mathcal{X}_{\Omega}(x,y)\dfrac{1}{\pi} \int\limits_{\Omega_x}{\dfrac{{\left[\Phi_{tilt}(x,y')-\Phi_{tilt}(x,y)\right]}k^{\{n,c,l\}}(y'-y)}{y'-y}\ dy'} \\
&=\mathcal{X}_{\Omega}(x,y)\dfrac{1}{\pi} \int\limits_{\Omega_x}{\dfrac{\left[b\left(y'-y\right)\right]k^{\{n,c,l\}}(y'-y)}{y'-y}\ dy'} \\
&=\mathcal{X}_{\Omega}(x,y)\dfrac{1}{\pi} \int\limits_{\Omega_x}{b\cdot k^{\{n,c,l\}}(y'-y)\ dy'} .
\end{align*}
Altogether, we obtain that tip is in the null space of $\boldsymbol R_y^{\{n,c,l\},lin}$ and tilt in the null space of  $\boldsymbol R_x^{\{n,c,l\},lin}$.
\end{proof}

\begin{figure}[!ht]
\centering
\includegraphics[scale = 0.8]{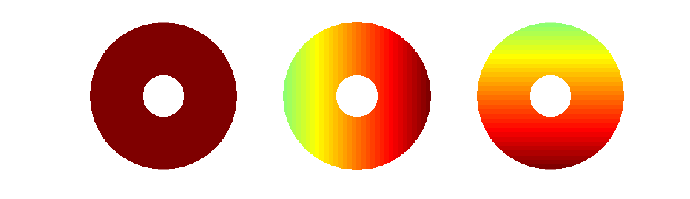}
\caption{The figures indicate piston, tip and tilt mode (from left to right).}
\label{fig:null_space}
\end{figure}

If we further simplify the roof sensor model by excluding the second term in \eqref{roof_operator} and consider the operators $\boldsymbol L^{\{n,c,l\}}$ defined in \eqref{3c.13}, we note that these operators are injective, i.e., $\mathcal{N}\left(\boldsymbol L^{\{n,c.l\}}\right) = \{0\}$. This assertion follows from the injectivity of the finite Hilbert transform shown in \cite{Engl97}. \bigskip

Note that the results of Proposition \ref{4.2} are directly transferred to the full non-linear roof sensor operators $\boldsymbol R^{\{n, c, l\}}$ discussed in part I of the paper \cite{HuSha18_1}. As soon as we consider functions including various powers of $x$ or $y$ in $\Phi_{tip/tilt}$, the corresponding measurements in $x$- or $y$-direction are equal to zero as well. The fact that $\boldsymbol R_{x}^{\{n,c,l\},lin}$ and $\boldsymbol R_{y}^{\{n,c,l\},lin}$ have different null spaces intensifies the requirement of an appropriate combination of the two data sets $s_x$ and $s_y$ for reconstruction methods which are based on the roof sensor model. \bigskip

One idea to appropriately combine both data sets is to reconstruct independently in both directions and average the reconstructions at the end as already considered for the CGNE, SD, and Landweber iteration above. However, it is not guaranteed that the final (averaged) solution $\Phi^{rec}$ fulfills both equations \eqref{5.1a}-\eqref{5.1b}. Another possibility is two consider $\left[\boldsymbol Q_x,\boldsymbol Q_y\right]$ as one single operator and a third one is to use a Kaczmarz strategy \cite{Ka37,Nat86} which is computationally cheaper and for which it is guaranteed that the equations \eqref{5.1a}-\eqref{5.1b} are fulfilled for the final solution. Kaczmarz methods, in general, have been developed for solving linear systems of equations. We have decided to implement Kaczmarz strategies for the pyramid sensor in combination with several of the above discussed algorithms.

\subsubsection{Landweber-Kaczmarz approach}

In practice, the Landweber algorithm is used because it is simple and each iteration is cheap. Though, the process usually requires a high number of iterations. Anyway, we do not experience slow convergence for reconstruction from pyramid data due to a close similarity between adjoint and inverse operators as investigated in \cite{Hut17} for the non-modulated sensor, i.e., for the finite Hilbert transform operator. When using proper basis functions for the representation of the incoming wavefront $\Phi$ and the measurements $s=\left[s_x,s_y\right]$ as derived in \eqref{5.6}, the involved operators can be precomputed offline. These facts make the Landweber iteration coupled with a Kaczmarz strategy interesting for wavefront reconstruction from pyramid sensor data. A general convergence analysis of the linear Landweber-Kaczmarz method can be found in \cite{Kow02}. \bigskip

In the linear setting the Landweber-Kaczmarz method for wavefront reconstruction from pyramid wavefront sensor measurements (Algorithm $4$) reads as: \bigskip\par

\begin{tabular}{l c r}
\hline
\textbf{Algorithm 4} Landweber-Kaczmarz iteration for pyramid sensors \\
\hline
choose $\Phi_0$, set relaxation parameters $\beta_1$, $\beta_2$ \\
for $i=1,\dots K$ do \\
\quad $\Phi_{i,0} = \Phi_{i-1}$ \\
\quad $ \Phi_{i,1} = \Phi_{i,0} + \beta_1 \boldsymbol Q_x^*\left(s_x-\boldsymbol Q_x\Phi_{i,0}\right)$\\ 
\quad $ \Phi_{i,2} = \Phi_{i,1} + \beta_2 \boldsymbol Q_y^*\left(s_y-\boldsymbol Q_y\Phi_{i,1}\right)$\\ 
\quad $\Phi_i=\Phi_{i,2}$ \\
endfor \\
$\Phi^{rec} = \Phi_K$ \\
\hline
\end{tabular}
\vspace*{1cm}


\subsubsection{Steepest descent-Kaczmarz approach}

The idea of modified steepest descent algorithms coupled with a Kaczmarz strategy is comparable to the method described in \cite{Cez08} for non-linear problems. As in the previous method, we cyclically consider each measurement equation \eqref{5.1a} and \eqref{5.1b}. \bigskip\par
Hence, for 
\be{5.8b}
J_x\left(\Phi\right) := \left|\left|\boldsymbol Q_x\Phi-s\right|\right|^2_{\mathcal{L}_2}, \qquad \qquad \qquad J_y\left(\Phi\right) := \left|\left|\boldsymbol Q_y\Phi-s\right|\right|^2_{\mathcal{L}_2},
\ee
the steepest descent-Kaczmarz (SD-K) method for wavefront reconstruction using pyramid sensors (Algorithm $5$a) is described by: \bigskip\par
\begin{tabular}{l c r}
\hline
\textbf{Algorithm 5a} Steepest descent-Kaczmarz method for pyramid sensors \\
\hline
choose $\Phi_0$ \\
for $i=1,\dots K$ do \\
\quad $\Phi_{i-1,0} = \Phi_{i-1}$ \\
\quad $d_{i-1,1} = -J_x'\left(\Phi_{i-1,0}\right)$ \\
\quad $ \tau_{i-1,1}= \min\limits_{t \in [0,\infty)} J_x\left(\Phi_{i-1,0}+td_{i-1,1}\right)$\\ 
\quad $ \Phi_{i-1,1} = \Phi_{i-1,0} + \tau_{i-1,1}d_{i-1,1}$\\ 
\quad $d_{i-1,2} = -J_y'\left(\Phi_{i-1,1}\right)$ \\
\quad $ \tau_{i-1,2}= \min\limits_{t \in [0,\infty)} J_y\left(\Phi_{i-1,1}+td_{i-1,2}\right)$\\ 
\quad $ \Phi_{i-1,2} = \Phi_{i-1,1} + \tau_{i-1,2}d_{i-1,2}$\\ 
\quad $\Phi_i=\Phi_{i-1,2}$ \\
endfor \\
$\Phi^{rec} = \Phi_K$ \\
\hline
\end{tabular} 
\vspace*{1cm}

During an observation, the reconstructions have to be repeated up to $0.3$ milliseconds.  Assuming that the incoming wavefront do not change much from one time steps to the next and, in particular, tip~\&~tilt do not change significantly, another idea (implemented in Algorithm $5$b) would be to reconstruct in $x$-direction for even time steps $t$ and in $y$-direction for the proximate odd time $t+1$ steps. The big advantage of Algorithm $5b$ consists in the reduction of the computational demand by more than $50\%$ compared to the normal SD approach.

\bigskip\par
\begin{tabular}{l c r}
\hline
\textbf{Algorithm 5b} Modified steepest descent-Kaczmarz method for pyramid sensors \\
\hline
if ($t$ mod $2 = 0$) do \\
\quad apply Algorithm $2$ in $x$-direction only \\
else if  \\
\quad apply Algorithm $2$ in $y$-direction only \\
endif \\
\hline
\end{tabular} 
\vspace*{1cm}

The post-loop step of Algorithm $1$-$3$, i.e., the averaging of the two reconstructions is not necessary for Algorithm $4$, $5$a, and $5$b since we only obtain one reconstruction $\Phi_K$. Please note that for the Kaczmarz-type methods it is merely necessary to choose one initial guess $\Phi_0$ instead of two as required for Algorithm~$1$-$3$.

\section{Numerical results}\label{chap:numerics}
We test the quality of the reconstruction approaches by continuously correcting the incoming wavefront in closed loop AO. In this setting, the wavefront sensor measures the incoming phase after passing the deformable mirror, i.e., the sensor sees the difference between the incoming wavefront and the correction induced by the mirror. For numerical simulations, we use the end-to-end simulation tool Octopus developed by ESO \cite{LeLouarn_OCTOPUS_04,LVK06}. As already mentioned in the introduction, we test the performance of the proposed methods for an ELT-sized telescope system. In particular, we consider the METIS instrument on the $39$~m sized ELT for non-modulated and modulated pyramid wavefront sensors having a $74\times 74$ spatial sampling. Although the observing facility has a primary mirror diameter of $39$~m, for METIS only the inner $37$~m are used.
The incoming wavefronts are simulated by a realization of the von Karman atmospheric model having $35$ layers. The system runs at a frequency of $1$~kHz for the non-modulated sensor and at a frequency of $500$~Hz for the modulated sensor. The mirror geometry in the simulations corresponds to the M4 geometry planned for the ELT which was just recently incorporated in Octopus. For the temporal control of the algorithms we use a simple integrator and optimize the gains with a resolution of $0.1$. 

As a quality measure we use the long-exposure (LE) Strehl ratio, which is computed as the average on-axis Strehl ratio for all performed time steps. The Strehl ratio is defined as the ratio of the peak aberrated image intensity from a point source compared to the maximum attainable intensity using an ideal optical system limited only by diffraction over the system's aperture. The maximum achievable value is $1$. 

The quality results of the algorithms are expressed in terms of long-exposure Strehl ratios at an observing wavelength of $2.2~\mu m$ (K-band). 
Note that according to the specifications of the METIS instrument, K-band is not included in the science range. Instead, observations are performed in L-band (at $\lambda_1=3.0~\mu m$, $\lambda_2=3.7~\mu m$), in M-band (at $\lambda=4.7~\mu m$) and in N-band (at $\lambda=10.0~\mu m$). For analysis purposes, however, we find it useful to have the output at a shorter wavelength as well. As such, we use $\lambda=2.2~\mu m$ in the K-band where the imaging is performed.

In our numerical tests we evaluate the reconstruction quality in a range of photon flux levels between 50 and 10000 photons per subaperture per frame for median atmospheric conditions. The simulation parameters are summarized in Table~\ref{table:1}. In order to speed up convergence to the closed loop, in the first 13 time steps we apply the CuReD reconstructor \cite{Ros11,Ros12}, which corrects mainly for the low frequencies in the wavefront.  \bigskip

\begin{table}
\renewcommand{\arraystretch}{1.2}
\begin{center}
\begin{small}
\begin{tabular}{ l   l  } 
 \hline
 \textbf{Simulation parameters} & \\
 \hline
 telescope diameter & $37$~m  \\
central obstruction & $30\%$ \\
science target & on-axis (SCAO)  \\ 
WFS & PWFS \\
 sensing band & K ($2.2$~$\mu$m) \\
evaluation bands & K ($2.2$~$\mu$m) \\
 & L ($3.0, 3.7$~$\mu$m) \\
 & M ($4.7$~$\mu$m) \\
 & N ($10.0$~$\mu$m) \\
modulation  & $[0,4]$ $\lambda/D$ \\
controller & integrator \\
atmospheric model & von Karman \\
number of simulated layers & $35$ \\
outer scale $L_0$ & $25$~m \\
atmosphere & median \\
Fried radius $r_0$ at $\lambda = 500$~nm & $0.157$~m \\
number of subapertures & $74 \times 74$\\
number of active subapertures & $[3912,4128]$ out of $5476$ \\
frame rate & $[1000,500]$~Hz  \\ 
DM delay & $1$ \\
detector read-out noise & $1$ electron/pixel \\
background flux & $0.000321$ photons/pixel/frame \\
photon flux & $[50,100,1000,10000]$ \\
iterations per simulation & $500$\\ 
 \hline
\end{tabular} \caption{Test case setting.}\label{table:1}
\end{small}
\end{center}
\end{table}


\subsection{Optimal step size choice for SD iteration in the context of WF reconstruction from pyramid data}

Before we compare the reconstruction quality of all proposed methods, we investigate the optimal step size choice for the steepest descent algorithm applied to WF reconstruction. For that analysis we consider the METIS instrument on the ELT having a pyramid sensor without modulation incorporated. The simulation parameters are identical to those listed in Table~\ref{table:1}. As photon flux, we use 10000 photons per subaperture per frame. The reconstruction quality is evaluated after 500 time steps using $5$ SD-iterations for each reconstruction in order to find the optimal choice of the step size. As listed in Table~\ref{table:step_sizes}, best results are obtained for the SD iteration combined with the classical steepest descent step size. The reason for the small number of performed iterations is (amongst others) related to the roof sensor approximation for modeling a pyramid sensor and discussed below in more detail. 

\begin{table}
\renewcommand{\arraystretch}{1.2}
\begin{center}
\begin{small}
\begin{tabular}{ l   c  } 
 \hline
 \textbf{step size choice} & \textbf{LE Strehl ratio} \\
 \hline
classical SD &  0.8322\\
minimal gradient & 0.8310 \\
Barzilai-Borwein 1 & 0.8311\\
Barzilai-Borwein 2 & 0.8316\\
Cauchy-Barzilai-Borwein 1 & 0.8317 \\
 \hline
\end{tabular} \caption{SD-reconstruction (Algorithm $2$) results for the non-modulated sensor in the K-band after 500 time steps using different step sizes.}\label{table:step_sizes}
\end{small}
\end{center}
\end{table}


\subsection{Simulated closed loop performance}
Let us analyze the closed loop performance of the developed algorithms and compare their reconstruction quality.
Our reconstruction methods are all based on a simplification of the full pyramid sensor model. As a consequence, after some iteration steps, the reconstructions suffer from an approximation error and depart from the true solution of the full pyramid sensor model although the residuals 
\begin{equation} \label{res_roof}
 \left|\left|s-\dfrac{1}{2} \ \boldsymbol R^{lin}\Phi_i\right|\right|
 \end{equation}
with respect to the simplified model continue to scale down during the iterations. Due to the fact that the full non-linear pyramid sensor model $\boldsymbol P$ consists of two terms $\boldsymbol P=\boldsymbol P^{lin}+\boldsymbol P^{rest}$, where the first term again contains two terms $\boldsymbol P^{lin}=\boldsymbol R^{lin}+\boldsymbol S^{lin}$ (see part I of the paper \cite{HuSha18_1} for more details), a reduction of the roof sensor residual \eqref{res_roof} can imply an error increase of $\boldsymbol S^{lin}\Phi + \boldsymbol P^{rest} \Phi$ in the residual corresponding to the full pyramid sensor model
\begin{equation*}
 \left|\left|s-\dfrac{1}{2}\boldsymbol P\Phi_i\right|\right| = \left|\left|s- \dfrac{1}{2}\boldsymbol R^{lin}\Phi_i -\dfrac{1}{2} \left( \boldsymbol S^{lin} + \boldsymbol P^{rest}\right)\Phi_i\right|\right| .
 \end{equation*}
Besides the approximation error, another error source, the data error, is present in the reconstruction process. It is inevitable to search for an adequate stopping criterion taking into account both the difference between the real pyramid sensor operator $\boldsymbol P$ providing the measurements $s$ and the approximate operator $\boldsymbol R^{lin}$, which builds the foundation of the model-based reconstruction algorithms, as well as data errors. For choosing the regularization parameter in the generally non-linear problem of wavefront reconstruction from pyramid data, we discuss the usage of Morozov's discrepancy principle. Assume that the pyramid sensor provides noisy data $s^\delta$ fulfilling $\left|\left|s-s^\delta\right|\right|< \delta$ for some noise level $\delta > 0$. The iteration is terminated with stopping index $k_*(\delta,s^\delta)$ when for the first time the residual is below $\tau\delta$ for some $\tau>1$, i.e., 
$$\left|\left|s^{\delta} - \dfrac{1}{2}\boldsymbol R^{lin}\Phi_i^\delta\right|\right| > \tau\delta \qquad 0\le i < k_* \qquad \qquad \text{and} \qquad \qquad \left|\left|s^\delta-\dfrac{1}{2}\boldsymbol R^{lin}\Phi_{k_*}^\delta\right|\right| \le \tau\delta.$$
The discrepancy principle combined with a criterion for controlling the approximation error can be transferred to the application of only a few CGNE- or SD-iterations resulting in a very low value for $k_*$ as confirmed by a huge number of numerical simulations performed within this study. In particular, one iteration suffices to provide high reconstruction quality when using a warm restart of the system. That is, in the first time step, the initial guess is chosen as zero, i.e., $\Phi_{0,0}:=0$, and at time step $t>0$ the initial value is set to the reconstructed phase of the previous step, i.e., $\Phi_{t,0}:=\Phi_{t-1}^{rec}$. By employing the reconstruction of the previous step as initial guess $\Phi_0$, we significantly decrease the computational complexity since Algorithm $1$, $2$ and $5$ are scaled down to non-iterative gradient based methods by applying only one corresponding iteration step. The warm restart technique improves the convergence speed of the iterative solvers and additionally slightly increases the quality performance. The suitable number of iterations is also depending on the number of incident photons, since a high photon flux results in reduced data noise and vice versa. \bigskip

In our applications a total number of $K=1$ iterations turned out to be optimal with respect to the reconstruction quality and the computational
complexity of the method. Except for the Landweber type approaches (Algorithm $3$ and $4$), we use more than one iteration, but already $K=5$ Landweber steps combined with an adapted choice of the relaxation parameter and the warm restart technique are enough to obtain satisfying reconstruction quality. 

In case of one CGNE- or SD-iteration the two algorithms coincide when using the classical steepest descent step size \eqref{tau_SD}. Additionally, the step sizes in the SD method discussed in the previous Section do not differ for one SD-iterate except for the classical SD step size and the MG step size. Hence, for the numerical simulations with results provided in Table~\ref{table:numerical_results} and Table~\ref{table:numerical_results4} we used the minimal gradient step size \eqref{tau_MG} in order to have an additional comparison of step size choices as well. Since Algorithm $5$b has a reduced computational complexity compared to Algorithm $5$a, we only consider the modified SD-Kaczmarz algorithm in our numerical tests. As above, one SD-Kaczmarz iteration suffices as well. \bigskip

Numerical tests suggest that for METIS an interpolation to a finer grid than given by the subaperture spacing results in an increased reconstruction quality. In the XAO case the corresponding improvement was less significant. This may be related to the difference in subaperture sizes of both systems ($21$~cm in XAO versus $50$~cm in METIS). \bigskip

Corresponding results having a cold start ($\Phi_0=0$) for every time step $t$ can be found in Table~\ref{table:step_sizes} for Algorithm $2$ utilizing a pyramid sensor without modulation while results using the warm restart technique for all presented algorithms are summarized in Table~\ref{table:numerical_results} for the non-modulated pyramid sensor and in Table~\ref{table:numerical_results4} for the modulated pyramid sensor. Hence, the warm restart technique improves the reconstruction quality of the SD approach from an LE Strehl ratio of $0.8322$ having a cold start to $0.8412$ with the warm restart. \bigskip

In case of zero modulation, best reconstruction quality is obtained for the Landweber-Kaczmarz approach using the two measurement sets alternating and for the gradient based approaches (Algorithms $1$-$2$) calculating two reconstructions and averaging at the end. For the sensor having modulation $4$~$\lambda/D$, surprisingly, the CGNE approach even outperforms the Landweber-Kaczmarz algorithm except for the simulations with $50$ photons per subapertures per frame. However, the differences in the results are very small anyhow. In addition to the K-band results shown in Table~\ref{table:numerical_results}, we provide the long-exposure Strehl ratios in other science bands as defined by the instrument specifications. Table~\ref{table:numerical_results_other_bands} shows the quality in L-, M-, and N-bands obtained with the Landweber-Kaczmarz algorithm in the high flux case (10000 ph/subaperture/frame) for the non-modulated sensor. \bigskip

The simulations for a modulated sensor whose results are presented in Table~\ref{table:numerical_results4} were performed with a frame rate of $500$~Hz. In order to have a direct comparison of the non-modulated and modulated sensor, we additionally run a simulation at a frame rate of $1$~kHz (instead of $500$~Hz) using a pyramid sensor with modulation $4$~$\lambda/D$. For the modulated sensor with the CGNE method we obtain the LE Strehl ratio of 0.8782 in the K-band in the high flux case after 500 time steps and the LE Strehl ratio of 0.8415 for the non-modulated sensor. This result fits well our previous experiences with other model-based reconstruction algorithms according to which the modulated sensor provides a higher quality compared to the non-modulated one. \bigskip

All in all, the developed reconstruction algorithms deliver comparable quality and allow for robust and accurate wavefront reconstruction with low computational costs.
\bigskip \par

\begin{table}
\renewcommand{\arraystretch}{1.2}
\begin{center}
\begin{small}
\begin{tabular}{ l   c c c c c } 
 \hline
 \textbf{photon flux} & \textbf{Algorithm 1} & \textbf{Algorithm 2} & \textbf{Algorithm 3} & \textbf{Algorithm 4} & \textbf{Algorithm 5b} \\
 \hline
50 &  0.8374& 0.8376 & 0.8332 & 0.8371 & 0.8331\\
100 & 0.8407 & 0.8409 & 0.8384 & 0.8415 &  0.8393 \\
1000 & 0.8414 & 0.8413 & 0.8395 &  0.8420 & 0.8412\\
10000 & 0.8415 & 0.8412 & 0.8396 & 0.8419 & 0.8413 \\
 \hline
\end{tabular} \caption{Long-exposure Strehl ratios in the K-band obtained with the presented algorithms after 500 closed loop simulation steps for a pyramid sensor without modulation. Best results are obtained for the CGNE approach (Algorithm $1$), the SD (Algorithm $2$), and Landweber-Kaczmarz iteration (Algorithm $4$).}\label{table:numerical_results}
\end{small}
\end{center}
\end{table}

\begin{table}
\renewcommand{\arraystretch}{1.2}
\begin{center}
\begin{small}
\begin{tabular}{ l   c c c c c } 
 \hline
 \textbf{photon flux} & \textbf{Algorithm 1} & \textbf{Algorithm 2} & \textbf{Algorithm 3} & \textbf{Algorithm 4} & \textbf{Algorithm 5b} \\
 \hline
50 &  0.8432& 0.8434 & 0.8427 & 0.8439 & 0.8340\\
100 & 0.8524 & 0.8520 & 0.8517 & 0.8510 &  0.8454 \\
1000 & 0.8597 & 0.8579 & 0.8590 &  0.8562 & 0.8570\\
10000 & 0.8604 & 0.8581 & 0.8595 & 0.8577 & 0.8580 \\
 \hline
\end{tabular} \caption{Long-exposure Strehl ratios in the K-band obtained with the presented algorithms after 500 closed loop simulation steps for a pyramid sensor with modulation. Here, the CGNE approach (Algorithm~$1$) provides the highest reconstruction quality in most of the cases. }\label{table:numerical_results4}
\end{small}
\end{center}
\end{table}

\begin{table}
\renewcommand{\arraystretch}{1.2}
\begin{center}
\begin{small}
\begin{tabular}{ l   c } 
 \hline
 \textbf{sensing wavelength} & \textbf{LE Strehl}  \\
 \hline
2.2 $\mu$m & 0.8419  \\ 
3.0 $\mu$m & 0.9107  \\
3.7 $\mu$m &  0.9401 \\
4.7 $\mu$m &  0.9623 \\
10.0 $\mu$m &  0.9915  \\
 \hline
\end{tabular} \caption{Long-exposure Strehl ratios in L-, M-, and N-bands obtained for the non-modulated pyramid sensor with the Landweber-Kaczmarz algorithm in the high flux case (10000 ph/subaperture/frame) after 500 closed loop simulation steps.}\label{table:numerical_results_other_bands}
\end{small}
\end{center}
\end{table}


\subsection{Comparison to interaction-matrix-based approaches}

In the literature, there are many variants of interaction-matrix-based approaches: statistical estimators or least-squares methods; zonal or modal \cite{Law_MVM_96} control approaches (i.e., the degrees of freedom are modes or actuators/subapertures). The least-squares approach applied to the pyramid wavefront sensor reaches high correction accuracy without regularization, at least if the number of degrees of freedom is small as demonstrated in \cite{Esposito2010_AO_for_LBT, Esposito_2011_pwfs_onsky,Esposito_ao4elt2_LBT_onsky}. However, for large-scale AO systems the least-squares attempt turned out to be less accurate and the minimum variance estimator allowing statistical regularization is preferred \cite{Bardsley_2008,Ellerbroek02} (see also MAP \cite{Clare_ao4elt2_2011,GaLo10,Law_MVM_96,Louarn_AO4ELT5}, MMSE \cite{Bardsley_2008}). One should note that regularization typically requires optimization (fine tuning) of the regularization parameters. Moreover, the regularized control matrix has to be recomputed each time the seeing conditions or the photon flux change, which is a rather time-consuming task. (As already mentioned, the computational complexity required for setting up the command matrix scales as $O(N^3)$ and the application of this command matrix on the sensor data as $O(N^2)$).

Often, in practice and also in simulations, the non-modulated sensor being operated with an interaction-matrix-based approach, is reported to be unstable, see, e.g, \cite{GaLo10,Louarn_AO4ELT5}. One can for instance apply some tricks, like using a \textquoteleft wrong\textquoteright \  command matrix derived for the modulated sensor, or heavily fine-tune the regularization parameters to filter out the unstable modes in the correct interaction matrix (measured or computed for the sensor with modulation 0), which has to be performed on the fly and is a very time-consuming task.

Recently, there was a result published in \cite{Louarn_AO4ELT5} for the non-modulated sensor running in Octopus with a modal MVM at $1$~kHz frame rate. The achieved quality in the K-band was reported to be 0.62 for the high flux case (10000 photons/subaperture/frame). For comparison, the pyramid sensor with modulation $4$~$\lambda/D$ was reported to provide in the same environment the LE Strehl ratio of 0.80.

As recently reported in \cite{Hippler_2018}, another variant of MVM, the zonal minimum variance reconstructor in the YAO simulation tool \cite{Rigaut_2013_ao4elt3_yao}, which is a zonal regularized approach, achieves LE Strehl of 0.89 in case of the modulated pyramid sensor (with modulation $4$) and high photon flux. \bigskip

Comparing the performances of the described algorithms, we can draw the following conclusions. For the pyramid sensor without modulation our reconstruction algorithms, which use the forward model of the sensor, allow not only to close the loop easily, but also to achieve a stable correction over time with a quality significantly higher compared to the interaction-matrix-based reconstructor utilized in~\cite{Louarn_AO4ELT5}. In case of the pyramid sensor with the optimal amount of modulation, our algorithms achieve a reconstruction quality which is slightly (only 0.012 points of LE Strehl for the simulations at a frame rate of $1$~kHz) below the best (known) result obtained with the zonal MMSE variant of MVM.

\section{Computational complexity} \label{chap:computational_complexity} 
We define the computational complexity of the algorithm as a number of required floating point operations (flops). Let $n$
denote the number of subapertures in  one direction, then $N = n^2$ indicates approximately the number of unknowns to be found.

\subsection{Complexity of Landweber iteration and Landweber-Kaczmarz iteration for pyramid sensors}

We only consider the complexity of the operations that have to be performed online and exclude the pre-calculations needed in the application of the operators $\boldsymbol Q$ and $\boldsymbol Q^*$ from our considerations. The number of floating point operations for every step in the Landweber iteration approach for wavefront reconstruction using pyramid sensors (Algorithm $3$ and Algorithm $4$) is provided in Table~\ref{table:PKI_FLOPs}. The post loop step of the Landweber algorithm consists of finding the average between the two resulting
reconstructions, which requires one summation and one division by a scalar. Altogether, this step is summed up to $2 n^2$ operations. \bigskip

\begin{table}
\renewcommand{\arraystretch}{1.2}
\begin{center}
\begin{small}
\begin{tabular}{ l| l   l  } 
& \textbf{operation} & \textbf{\# of flops} \\
 \hline
loop & $\boldsymbol Q_x \Phi_x$ & $2n^3-n^2$  \\
& $s_x-\boldsymbol Q_x\Phi_x$ & $n^2$ \\
& $\boldsymbol Q_x^*\left(s_x-\boldsymbol Q_x\Phi_x\right)$ & $2n^3-n^2$ \\
&$\beta\boldsymbol Q_x^*\left(s_x-\boldsymbol Q_x\Phi_x\right)$ & $n^2$ \\
&$\Phi_x+\beta\boldsymbol Q_x^*\left(s_x-\boldsymbol Q_x\Phi_x\right)$ & $n^2$ \\
post loop step & $\Phi = \tfrac{1}{2}\left(\Phi_x+\Phi_y\right)$ & $2n^2$ \\
\hline
\end{tabular} \caption{The number of flops to be performed online in the Landweber and Landweber-Kaczmarz method.}\label{table:PKI_FLOPs}
\end{small}
\end{center}
\end{table}

Since we perform the mentioned operations twice (in $x$- and in $y$-direction), for $K$ iterations we obtain the complexity $$C_{4}(n; K)=\left(8n^3+2n^2\right)\cdot K$$ for the Landweber-Kaczmarz approach (Algorithm $4$) and $$C_{3}(n; K)=\left(8n^3+2n^2\right)\cdot K + 2n^2$$ flops for the application of the Landweber iteration (Algorithm $3$) having the additional step of averaging. 


\subsection{Complexity of SD and SD-Kaczmarz algorithm for pyramid sensors}
We again only consider the operations performed online. The complexity of one steepest descent iteration (Algorithm $2$ and Algorithm $5$) is indicated in Table~\ref{table:SD_FLOPs}. In case of the classical steepest descent iteration (Algorithm $2$) a subsequent averaging  $\Phi=1/2(\Phi_x+\Phi_y)$ has to be performed, which costs additionally $2n^2$ flops. Therefore, the number of flops for the steepest descent-Kaczmarz approach applied to pyramid sensors (Algorithm $5$a) is given by 
\begin{align*}
C_{5a}(n; K)&=\left(12n^3+8n^2+2\right)\cdot K ,
\end{align*}
for the modified Algorithm $5$b by 
\begin{align*}
C_{5b}(n; K)&=\left(6n^3+4n^2+1\right)\cdot K ,
\end{align*}
and for the steepest descent approach (Algorithm $2$) by $$C_{2}(n; K)=\left(12n^3+8n^2+2\right)\cdot K + 2n^2,$$ where $K$ indicates the number of steepest descent steps.

\begin{table}
\renewcommand{\arraystretch}{1.2}
\begin{center}
\begin{small}
\begin{tabular}{l | l   l  } 

& \textbf{operation} & \textbf{\# of flops} \\
 \hline
loop & $J'(\Phi_x)$ & $4n^3-n^2$  \\
 &evaluation of $\tau_{SD}$ & $2n^3+3n^2+1$ \\
& $\Phi_x-\tau J'(\Phi_x)$ & $2n^2$ \\
post loop step & $\Phi = \dfrac{\Phi_x+\Phi_y}{2}$ & $2n^2$ \\
\hline
\end{tabular} \caption{The number of flops in the steepest descent and steepest descent-Kaczmarz method for pyramid sensors.}\label{table:SD_FLOPs}
\end{small}
\end{center}
\end{table}


\subsection{Complexity of CGNE for pyramid sensors}

The CGNE method consists of three steps:
\begin{enumerate}
\item a pre-computation and initialization step which have to be done for both $x$- and $y$-direction once,
\item the CG-loop for $K$ iterations performed twice in $x$- and $y$-direction,
\item a post loop step in which we average the two obtained reconstructions.
\end{enumerate}

The number of flops in the CGNE algorithm for pyramid sensors for operations which are not pre-computed offline is indicated in Table~\ref{table:CG_FLOPs}. Summing up the specified operations for both $s_x$ and $s_y$ data, we see that the initialization step consists of $8n^3$ flops, the loop of $\left(8n^3+20n^2+4\right)\cdot K$ and the post loop step of $2n^2$ operations. Hence, the CGNE complexity for pyramid sensors sums up as $$C_1(n;K)=\left(8n^3+20n^2+4\right)\cdot K + 8n^3 + 2n^2.$$
However, since the CG method is known to require the fewest number of iterations, $K$ usually is smaller compared to, e.g., the Landweber iteration.
\begin{table}
\renewcommand{\arraystretch}{1.2}
\begin{center}
\begin{small}
\begin{tabular}{l | l   l  } 
& \textbf{operation} & \textbf{\# of flops} \\
 \hline
 init & computation of $d_{x,0}$ & $2n^3$  \\
& computation of $s_{x,0}$ & $2n^3-n^2$ \\
& initialization of $p_{x,1}$ & $n^2$ \\
loop & computation of $q_x$ & $2n^3-n^2$ \\
& computation of $\alpha$ & $4n^2+1$ \\
& computation of $\Phi_x$ & $2n^2$ \\
& computation of $d_x$ & $2n^2$ \\
& computation of $s_x$ & $2n^3-n^2$ \\
& computation of $\beta$ & $2n^2+1$ \\
& computation of $p_x$ & $2n^2$ \\
post loop step & $\Phi = \tfrac{1}{2}\left(\Phi_x+\Phi_y\right)$ & $2n^2$ \\
\hline
\end{tabular} \caption{The number of flops in the CGNE algorithm for pyramid sensors considering operations which are not precomputed offline.}\label{table:CG_FLOPs}
\end{small}
\end{center}
\end{table}


\subsection{Comparison to MVM}

In our notations, the complexity of standard MVM methods scales as $\mathcal{O}(N^2)=\mathcal{O}(n^4)$. For these studies, we calculate the reconstructions (mirror actuator commands) at the corners of the subapertures, and thus need to
consider approximately $n'=n+1$ phase values. For, e.g., $n=200$ subapertures, the complexity of MVM is roughly given by
$$ C_{MVM} (200) = 201^4 \approx 16 \cdot 10^8 = 1600 \cdot 10^6 . $$
The complexities of the developed methods are estimated in Table~\ref{table:complexity}. Here, we assume that the mirror actuators are equidistantly spaced on a squared shape although this is not employed in practice since not all actuators are actively controlled. However, for a theoretical comparison of complexities such assumptions are still relevant. Note that in principle $C_1(n;1)=C_2(n;1)$ because the CGNE algorithm can already be terminated after the calculation of $\Phi_1$. However, we consider one full CGNE step in the Table. As a remark, we mention that in our comparison we have omitted the additional computational effort required in the presented model-based algorithms for computation of deformable mirror commands from the reconstructed wavefront shape. This step can be represented as a bilinear interpolation from the $n\times n$ grid of subapertures to the $(n+1)\times (n+1)$ grid of DM actuators, which requires $4(n+1)^2$ flops to be performed. Also, we would like to mention that the time-saving features of MVM approaches like parallelizability and pipelineability are valid in our algorithms as well.

\begin{table}
\renewcommand{\arraystretch}{1.2}
\begin{center}
\begin{small}
\begin{tabular}{  l    c | c  c   c  |c } 
 \textbf{approach} & \textbf{complexity} & \textbf{flops} & \textbf{METIS (n=74)} &  \textbf{XAO (n=200)} & \textbf{XAO in $\%$}\\
 \hline
MVM & $\mathcal {O}\left(N^2\right)$ & $C_{MVM}(n)$   & $30 \cdot 10^6$ & $1600 \cdot 10^6$ & 100 \% \\
CGNE &$\mathcal {O}\left(N^{3/2}\right)$ & $C_{1}(n;1)$   & $6,9 \cdot 10^6$ & $130,9 \cdot 10^6$ & 8  \%\\
SD & $\mathcal {O}\left(N^{3/2}\right)$ &$C_{2}(n;1)$   & $5,1 \cdot 10^6$ & $97,9 \cdot 10^6$ & 6  \%\\
SD-Kaczmarz& $\mathcal {O}\left(N^{3/2}\right)$ &$C_{5a}(n;1)$   & $5,1 \cdot 10^6$ & $97,8 \cdot 10^6$ & 6  \% \\
modified SD-Kaczmarz& $\mathcal {O}\left(N^{3/2}\right)$ &$C_{5b}(n;1)$   & $2,6 \cdot 10^6$ & $48,9 \cdot 10^6$ & 3  \%\\
Landweber iteration& $\mathcal {O}\left(N^{3/2}\right)$ &$C_{3}(n;5)$   & $16,9 \cdot 10^6$ & $325,3 \cdot 10^6$ & 20  \% \\
Landweber-Kaczmarz it.& $\mathcal {O}\left(N^{3/2}\right)$ &$C_{4}(n;5)$   & $16,3 \cdot 10^6$ & $325,2 \cdot 10^6$ & 20  \% \\
 \hline
\end{tabular} \caption{The computational complexities of the algorithms analyzed in this paper compared to the implementation of an MVM method. Estimates of the number of flops necessary for the METIS instrument having a $74\times 74$ pyramid sensor and for an XAO system with a $200\times 200$ pyramid sensor. The last column demonstrates the computational effort of the new algorithms as percentage of the MVM effort for the XAO system.}\label{table:complexity}
\end{small}
\end{center}
\end{table}

The developed algorithms allow to significantly reduce the numerical effort of the wavefront reconstruction step in an AO loop compared to the computational load related to the solvers based on matrix-vector multiplication. This is illustrated especially well for the XAO system having a huge number of active actuators. The computational effort of MVM-based wavefront estimators is extremely demanding in this case. In contrast, the usage of analytically developed wavefront reconstructors allows one to heavily reduce the numerical effort of the AO loop. For instance as shown in Table~\ref{table:complexity}, the modified steepest descent algorithm reduces the computational load of the wavefront reconstruction step in the XAO loop to approximately $3 \%$ of the MVM effort while still providing high reconstruction quality.

\section*{Conclusion and Outlook}

In this paper we have studied the application of well-known iterative algorithms for solving the inverse problem of wavefront reconstruction from pyramid wavefront sensor data in the field of astronomical Adaptive Optics. From the performed end-to-end numerical simulations we can draw the conclusion that all studied algorithms deliver very similar reconstruction quality. However, it is preferable to apply the Kaczmarz versions of the algorithms or the CGNE approach, since they provide a slightly better reconstruction quality, though, the difference in the achieved quality between all the methods is minor. The best quality is obtained with the CGNE approach (Algorithm $1$) and with the Landweber-Kaczmarz iteration (Algorithm $4$), which at the same time is part of the slowest among the algorithms under comparison. If one decides to go for speed at the cost of a negligible quality loss, one should choose the modified steepest descent-Kaczmarz version combined with the classical step size choice (Algorithm 5b). \bigskip

As shown by numerical results presented in this study, the proposed algorithms, which are partially iterative methods, allow to keep the numerical effort of the wavefront reconstruction step in an AO loop low compared to the computational load of solvers based on matrix-vector-multiplication. This has an especially big impact for the considered XAO system having a huge number of active actuators. For instance, the modified steepest descent algorithm reduces the computational load of the wavefront reconstruction step in the XAO loop to approximately $3 \%$ of the MVM effort while still providing high reconstruction quality.

Even when using simplifications of the pyramid sensor model, all proposed algorithms provide stable high-quality reconstruction and (almost) reach the quality of interaction-matrix-based approaches in which the full pyramid model is assumed. Especially for the non-modulated sensor, the linear iterative algorithm give stable and very accurate wavefront reconstructions. If we compare the methods presented in this paper with the P-CuReD, all of them are outmatched by the P-CuReD with respect to both speed and quality. Nevertheless, in the proposed iterative methods, there is a possibility to investigate the full pyramid sensor model for future developments. Remarkable quality improvements are hoped for those adaptions. A big advantage of the iterative methods over the P-CuReD is that the full pyramid sensor model or real life features such as telescope spiders or the low wind effect can be incorporated. For the P-CuReD it may even be impossible to adapt the algorithm to a more sophisticated pyramid sensor model. \bigskip



Finally, we would like to mention that investigations of the behavior of iterative algorithms in the presence of the so called \grqq spiders\grqq \ (support structures of the secondary mirror segmenting the telescope pupil into disjoint parts), the application of non-linear iterative algorithms for wavefront recnstruction as well as further quality evaluations to meet specifications of the METIS instrument are part of our further research~\cite{Hut18,HuShaOb18,ObRafShaHu18_proc}. 

\section*{Acknowledgements}

This work has been partly supported by the Austrian Federal Ministry of Science and Research (HRSM) and the Austrian Science Fund (F68-N36, project 5).

\bibliographystyle{plain}
\bibliography{arXiv_pyramid_iterativeII}

\begin{thebibliography}{10}

\bibitem{Bardsley_2008}
J.~M. Bardsley.
\newblock Wavefront reconstruction methods for adaptive optics systems on
  ground-based telescopes.
\newblock {\em SIAM Journal on Matrix Analysis and Applications}, 30:67--83,
  2008.

\bibitem{Bar88}
J.~Barzilai and J.~Borwein.
\newblock Two-point step size gradient methods.
\newblock {\em IMA Journal of Numerical Analysis}, 8:141--148, 1988.

\bibitem{Bitenc_cured_onsky}
U.~Bitenc, A.~Basden, N.~A. Bharmal, T.~Morris, N.~Dipper, E.~Gendron,
  F.~Vidal, D.~Gratadour, G.~Rousset, and R.~Myers.
\newblock {On-sky tests of the CuReD and HWR fast wavefront reconstruction
  algorithms with CANARY}.
\newblock {\em Monthly Notices of the Royal Astronomical Society},
  448(2):1199--1205, 2015.

\bibitem{bitenc_ao4elt3_onsky}
{U.} {Bitenc}, {M.} {Rosensteiner}, {N.} {Bharmal}, {A.} {Basden}, {T.}
  {Morris}, {A.} {Obereder}, {N.} {Dipper}, {E.} {Gendron}, {F.} {Vidal}, {G.}
  {Rousset}, {D.} {Gratadour}, {O.} {Martin}, {Z.} {Hubert}, and {R.} {Myers}.
\newblock {Tests of novel wavefront reconstructors on sky with CANARY}.
\newblock In {\em Proceedings of the Third AO4ELT Conference}, 2013.

\bibitem{Bra87}
H.~Brakhage.
\newblock On ill-posed problems and the method of conjugate gradients.
\newblock In H.~W. Engl and C.W. Groetsch, editors, {\em Inverse and Ill-Posed
  Problems}, pages 165 -- 175. Academic Press, 1987.

\bibitem{METIS_archiv}
B.~R. Brandl, M.~Feldt, A.~Glasse, M.~Guedel, S.~Heikamp, M.~Kenworthy,
  R.~Lenzen, M.~R. Meyer, F.~Molster, S.~Paalvast, E.~J. Pantin, S.~P. Quanz,
  E.~Schmalzl, R.~Stuik, L.~Venema, C.~Waelkens, and the NOVA-Astron
  Instrumentation~Group.
\newblock {METIS: the Mid-infrared E-ELT Imager and Spectrograph}.
\newblock https://arxiv.org/pdf/1409.3087.pdf.

\bibitem{BuDa06}
A.~Burvall, E.~Daly, S.~R. Chamot, and C.~Dainty.
\newblock Linearity of the pyramid wavefront sensor.
\newblock {\em Optics Express}, 14 (25):11925--11934, 2006.

\bibitem{Cauchy1847}
A.~Cauchy.
\newblock M\'{e}thode g\'{e}n\'{e}rale pour la r\'{e}solution des systemes
  d'\'{e}quations simultan\'{e}es.
\newblock {\em Comp. Rend. Sci. Paris}, (25):536--538, 1847.

\bibitem{Cez08}
A.~De Cezaro, M.~Haltmeier, A.~Leitao, and O.~Scherzer.
\newblock On steepest-descent-kaczmarz methods for regularizing systems of
  nonlinear ill-posed equations.
\newblock {\em Applied Mathematics and Computation}, 202(2):596 -- 607, 2008.

\bibitem{Chamot06}
S.~R. Chamot, C.~Dainty, and S.~Esposito.
\newblock Adaptive optics for ophthalmic applications using a pyramid wavefront
  sensor.
\newblock {\em Optics Express}, 14(2):518--526, 2006.

\bibitem{Clare_ao4elt2_2011}
R.~Clare and M.~Le Louarn.
\newblock {Numerical simulations of an Extreme AO system for an ELT}.
\newblock In {\em Proc. AO4ELT2}, pages 1100--1107, 2011.

\bibitem{Clenet_SPIE_2016_micado_pwfs}
Y.~Cl\'{e}net, T.~Buey, G.~Rousset, E.~Gendron, S.~Esposito, Z.~Hubert,
  L.~Busoni, M.~Cohen, A.~Riccardi, F.~Chapron, M.~Bonaglia, A.~Sevin,
  P.~Baudoz, P.~Feautrier, G.~Zins, D.~Gratadour, F.~Vidal, F.~Chemla,
  F.~Ferreira, N.~Doucet, S.~Durand, A.~Carlotti, C.~Perrot, L.~Schreiber,
  M.~Lombini, P.~Ciliegi, E.~Diolaiti, J.~Schubert, and R.~Davies.
\newblock {Joint MICADO-MAORY SCAO mode: specifications, prototyping,
  simulations and preliminary design}.
\newblock In {\em {Proc. SPIE 9909, Adaptive Optics Systems V}}, pages
  99090A--99090A--12, 2016.

\bibitem{Dai03_1}
Y.-H. Dai.
\newblock Alternate step gradient method.
\newblock {\em Optimization}, 52(4-5):395--415, 2003.

\bibitem{Dai03}
Y.-H. Dai and Y.-X. Yuan.
\newblock Alternate minimization gradient method.
\newblock {\em IMA Journal of Numerical Analysis}, 23(3):377--393, 2003.

\bibitem{Daly_2010}
E.~M. Daly and C.~Dainty.
\newblock Ophthalmic wavefront measurements using a versatile pyramid sensor.
\newblock {\em Appl. Opt.}, 49(31):G67--G77, 11 2010.

\bibitem{Alvarez_Thesis}
C.~A. Diez.
\newblock {\em {A 3-sided Pyramid Wavefront Sensor Controlled by a Neural
  Network for Adaptive Optics to reach diffraction-limited Imaging of the
  Retina}}.
\newblock PhD thesis, 2006.

\bibitem{Ellerbroek02}
B.~L. Ellerbroek.
\newblock Efficient computation of minimum-variance wave-front reconstructors
  with sparse matrix techniques.
\newblock {\em JOSA A}, 19(9):1803--1816, 2002.

\bibitem{Engl97}
H.~W. Engl.
\newblock {\em Integralgleichungen}.
\newblock Springer Wien, 1997.

\bibitem{Engl}
H.~W. Engl, M.~Hanke, and A.~Neubauer.
\newblock {\em Regularization of Inverse Problems}.
\newblock Kluwer Academic Publishers, Dordrecht, Boston, London, 2000.

\bibitem{engler_spie_2018}
B.~Engler, S.~Weddell, M.~Le Louarn, and R.~Clare.
\newblock Effects of the telescope spider on extreme adaptive optics systems
  with pyramid wavefront sensors.
\newblock In {\em Proc. SPIE}, volume 10703, pages 10703 -- 107035F -- 13,
  2018.

\bibitem{Esposito_05}
S.~Esposito, E.~Pinna, A.~Puglisi, A.~Tozzi, and P.~Stefanini.
\newblock Pyramid sensor for segmented mirror alignment.
\newblock {\em Optics Letters}, 30(19):2572--2574, 10 2005.

\bibitem{Esposito_2012_pwfs_NGS_SCAO_GMT}
S.~Esposito, E.~Pinna, F.~Quir\'{o}s-Pacheco, A.~T. Puglisi, L.~Carbonaro,
  M.~Bonaglia, V.~Biliotti, R.~Briguglio, G.~Agapito, C.~Arcidiacono,
  L.~Busoni, M.~Xompero, A.~Riccardi, L.~Fini, and A.~Bouchez.
\newblock {Wavefront sensor design for the GMT natural guide star AO system}.
\newblock In {\em {Proc. SPIE 8447, Adaptive Optics Systems III}}, pages
  84471L--84471L--10, 2012.

\bibitem{Esposito_ao4elt2_LBT_onsky}
S.~{Esposito}, A.~{Riccardi}, L.~{Fini}, E.~{Pinna}, A.~{Puglisi}, F.~{Quiros},
  M.~{Xompero}, R.~{Briguglio}, L.~{Busoni}, P.~{Stefanini}, C.~{Arcidiacono},
  G.~{Brusa}, and D.~{Miller}.
\newblock {LBT AO on-sky results}.
\newblock In {\em {Proceedings of the Second AO4ELT Conference}}, page~3, 09
  2011.

\bibitem{Esposito_2011_pwfs_onsky}
S.~Esposito, A.~Riccardi, E.~Pinna, A.~Puglisi, F.~Quir\'{o}s-Pacheco,
  C.~Arcidiacono, M.~Xompero, R.~Briguglio, G.~Agapito, L.~Busoni, L.~Fini,
  J.~Argomedo, A.~Gherardi, G.~Brusa, D.~Miller, J.~C. Guerra, P.~Stefanini,
  and P.~Salinari.
\newblock {Large Binocular Telescope Adaptive Optics System: new achievements
  and perspectives in adaptive optics}.
\newblock In {\em {Proc. SPIE 8179, SAR Image Analysis, Modeling, and
  Techniques XI}}, pages 814902--814902--10, 2011.

\bibitem{Esposito2010_AO_for_LBT}
S.~Esposito, A.~Riccardi, F.~Quir\'{o}s-Pacheco, E.~Pinna, A.~Puglisi,
  M.~Xompero, R.~Briguglio, L.~Busoni, L.~Fini, P.~Stefanini, G.~Brusa,
  A.~Tozzi, P.~Ranfagni, F.~Pieralli, J.~C. Guerra, C.~Arcidiacono, and
  P.~Salinari.
\newblock {Laboratory characterization and performance of the high-order
  adaptive optics system for the Large Binocular Telescope}.
\newblock {\em Applied Optics}, 49(31):G174--G189, 11 2010.

\bibitem{Fusco_2010_spie_atlas_eelt}
T.~Fusco, S.~Meimon, Y.~Clenet, M.~Cohen, H.~Schnetler, J.~Paufique, V.~Michau,
  J.-P. Amans, D.~Gratadour, C.~Petit, C.~Robert, P.~Jagourel, E.~Gendron,
  G.~Rousset, J.-M. Conan, and N.~Hubin.
\newblock {ATLAS: the E-ELT laser tomographic adaptive optics system}.
\newblock In {\em {Proc. SPIE 7736, Adaptive Optics Systems II}}, pages
  77360D--77360D--12, 2010.

\bibitem{Fusco_2010_spie_harmoni_eelt}
T.~Fusco, N.~Thatte, S.~Meimon, M.~Tecza, F.~Clarke, and M.~Swinbank.
\newblock {Adaptive optics systems for HARMONI: a visible and near-infrared
  integral field spectrograph for the E-ELT}.
\newblock In {\em {Proc. SPIE 7736, Adaptive Optics Systems II}}, pages
  773633--773633--12, 2010.

\bibitem{GaLo10}
A.~Garcia-Rissmann and M.~{Le Louarn}.
\newblock Scao simulation results with a pyramid sensor on an elt-like
  telescope.
\newblock In {\em 1st AO4ELT Conference - Adaptative Optics for Extremely Large
  Telescopes proceedings}, 2010.

\bibitem{Gil77}
S.~F. Gilyazov.
\newblock Iterative solution methods for inconsistent operator equations.
\newblock {\em Moscow Univ. Comput. Math. Cybernet}, 3:78--84, 1977.

\bibitem{Hanke95}
M.~Hanke.
\newblock Conjugate gradient type methods for ill-posed problems.
\newblock {\em Longman Scientific \& Technical, Harlow, Essex}, 1995.

\bibitem{Hes52}
M.~R. Hestenes and E.~Stiefel.
\newblock Methods of conjugate gradients for solving linear systems.
\newblock {\em Journal of Research of the National Bureau of Standards},
  49(6):409--436, 1952.

\bibitem{Hippler_2018}
S.~Hippler, M.~Feldt, T.~Bertram, W.~Brandner, F.~Cantalloube, B.~Carlomagno,
  O.~Absil, A.~Obereder, Iu. Shatokhina, and R.~Stuik.
\newblock {Single conjugate adaptive optics for the ELT instrument METIS},
  2017.
\newblock submitted.

\bibitem{Hu18_thesis}
V.~Hutterer.
\newblock {\em {Model-based wavefront reconstruction approaches for pyramid
  wavefront sensors in Adaptive Optics}}.
\newblock PhD thesis, Johannes Kepler University Linz, 2018.

\bibitem{Hut18}
V.~Hutterer and R.~Ramlau.
\newblock {Non-linear wavefront reconstruction methods for pyramid sensors
  using Landweber and Landweber-Kaczmarz iteration}.
\newblock {\em {Applied Optics}}, page to be published, 2018.

\bibitem{Hut17}
V.~Hutterer and R.~Ramlau.
\newblock {Wavefront reconstruction from non-modulated pyramid wavefront sensor
  data using a singular value type expansion}.
\newblock {\em Inverse Problems}, 34(3):035002, 2018.

\bibitem{HuSha18_1}
V.~Hutterer, R.~Ramlau, and Iu. Shatokhina.
\newblock {Real-time Adaptive Optics with pyramid wavefront sensors: A
  theoretical analysis of the pyramid sensor model}, 2018.
\newblock submitted.

\bibitem{HuShaOb18}
V.~Hutterer, Iu. Shatokhina, A.~Obereder, and R.~Ramlau.
\newblock {Advanced reconstruction methods for segmented ELT pupils using
  pyramid sensors}, 2018.
\newblock submitted.

\bibitem{ShatHut_spie2018_overview}
V.~Hutterer, Iu. Shatokhina, A.~Obereder, and R.~Ramlau.
\newblock {Wavefront reconstruction for ELT-sized telescopes with pyramid
  wavefront sensors}.
\newblock In {\em Proc. SPIE}, volume 10703, pages 10703 -- 1070344 -- 18,
  2018.

\bibitem{Ig11}
I.~Iglesias.
\newblock Pyramid phase microscopy.
\newblock {\em Optics Letters}, 36(18):3636--3638, 2011.

\bibitem{Iglesias02}
I.~Iglesias, R.~Ragazzoni, Y.~Julien, and P.~Artal.
\newblock Extended source pyramid wave-front sensor for the human eye.
\newblock {\em Optics Express}, 10(9):419--428, 2002.

\bibitem{Ig13}
I.~Iglesias and F.~Vargas-Martin.
\newblock Quantitative phase microscopy of transparent samples using a liquid
  crystal display.
\newblock {\em Journal of Biomedical Optics}, 18(2):026015--1--5, 2013.

\bibitem{Ka37}
S.~Kaczmarz.
\newblock {Angen\"{a}herte Aufl\"{o}sung von Systemen linearer Gleichungen}.
\newblock {\em Bulletin International de l'Acad\'{e}mie Polonaise des Sciences
  et des Lettres}, 35:355--357, 1937.

\bibitem{KaNa72}
W.~J. Kammerer and M.~Z. Nashed.
\newblock On the convergence of the conjugate gradient method for singular
  linear operator equations.
\newblock {\em SIAM Journal on Numerical Analysis}, 9(1):165--181, 1972.

\bibitem{KoVe10}
V.~Korkiakoski and C.~V\'{e}rinaud.
\newblock {Extreme adaptive optics simulations for EPICS}.
\newblock In {\em {Proceedings of the First AO4ELT Conference}}, page 03007,
  2010.

\bibitem{Kow02}
R.~Kowar and O~Scherzer.
\newblock Convergence analysis of a landweber-kaczmarz method for solving
  nonlinear ill-posed problems.
\newblock {\em Journal of Ill-Posed and Inverse Problems}, 23, 01 2002.

\bibitem{Landweber51}
L.~Landweber.
\newblock An iteration formula for fredholm integral equations of the first
  kind.
\newblock {\em American Journal of Mathematics}, 73(3):615--624, 1951.

\bibitem{Law_MVM_96}
N.~F. Law and R.~G. Lane.
\newblock {Wavefront estimation at low light levels}.
\newblock {\em Optics Communications}, 126:19--24, 1996.

\bibitem{LeLouarn_OCTOPUS_04}
M.~Le~Louarn, C.~Verinaud, V.~Korkiakoski, and E.~Fedrigo.
\newblock {Parallel simulation tools for AO on ELTs}.
\newblock In {\em Advancements in Adaptive Optics}, Proc. SPIE 5490, pages
  705--712, 2004.

\bibitem{LVK06}
M.~{Le Louarn}, C.~V\'{e}rinaud, V.~Korkiakoski, N.~Hubin, and E.~Marchetti.
\newblock {Adaptive optics simulations for the European Extremely Large
  Telescope - art. no. 627234}.
\newblock In {\em {Advances in Adaptive Optics II, Prs 1-3}}, volume 6272,
  pages {U1048--U1056}, {2006}.

\bibitem{Louarn_AO4ELT5}
M.~Le Louarn, P.-Y. Madec, S.~Oberti, J.~Paufique, M.~Sarazin, S.~Stroebele,
  and M.~Esselborn.
\newblock {Latest AO simulation results for the E-ELT}, 2017.
\newblock Poster presented at AO4ELT5 Conference.

\bibitem{Louis89}
A.K. Louis.
\newblock {\em {Inverse und schlecht gestelle Probleme}}.
\newblock B.G: Teubner Stuttgart, 1989.

\bibitem{Macintosh_2006_spie_pwfs_xao_tmt}
B.~Macintosh, M.~Troy, R.~Doyon, J.~Graham, K.~Baker, B.~Bauman, C.~Marois,
  D.~Palmer, D.~Phillion, Lisa Poyneer, I.~Crossfield, P.~Dumont, B.~M. Levine,
  M.~Shao, G.~Serabyn, C.~Shelton, G.~Vasisht, J.~K. Wallace, J.-F. Lavigne,
  P.~Valee, N.~Rowlands, K.~Tam, and D.~Hackett.
\newblock Extreme adaptive optics for the thirty meter telescope.
\newblock In {\em {Proc. SPIE 6272, Advances in Adaptive Optics II}}, pages
  62720N--62720N--15, 2006.

\bibitem{Mieda_spie_2016_pwfs_truth_TMT}
E.~Mieda, M.~Rosensteiner, M.~van Kooten, J.-P. V\'{e}ran, O.~Lardiere, and
  G.~Herriot.
\newblock {Testing the pyramid truth wavefront sensor for NFIRAOS in the lab}.
\newblock In {\em {Proc. SPIE 9909, Adaptive Optics Systems V}}, pages
  99091J--99091J--10, 2016.

\bibitem{Nat86}
F.~Natterer.
\newblock {\em The Mathematics of Computerized Tomography}.
\newblock Vieweg+Teubner Verlag, 1986.

\bibitem{Neichel_2016_spie_pwfs_harmoni_eelt}
B.~Neichel, T.~Fusco, J.-F. Sauvage, C.~Correia, K.~Dohlen, K.~El-Hadi,
  L.~Blanco, N.~Schwartz, F.~Clarke, N.~A. Thatte, M.~Tecza, J.~Paufique,
  J.~Vernet, M.~Le Louarn, P.~Hammersley, J.-L. Gach, S.~Pascal, P.~Vola,
  C.~Petit, J.-M. Conan, A.~Carlotti, C.~V\'{e}rinaud, H.~Schnetler, I.~Bryson,
  T.~Morris, R.~Myers, E.~Hugot, A.~M. Gallie, and David~M. Henry.
\newblock {The adaptive optics modes for HARMONI: from Classical to Laser
  Assisted Tomographic AO}.
\newblock In {\em {Proc. SPIE 9909, Adaptive Optics Systems V}}, pages
  990909--990909--15, 2016.

\bibitem{ObRafShaHu18_proc}
A.~Obereder, R.~Raffetseder, Iu. Shatokhina, and V.~Hutterer.
\newblock {Dealing with spiders on ELTs: using a Pyramid WFS to overcome
  residual piston effects}.
\newblock In {\em Proc. SPIE}, volume 10703, pages 10703 -- 107031D -- 19,
  2018.

\bibitem{Phill06}
D.~W. Phillion and K.~Baker.
\newblock Two-sided pyramid wavefront sensor in the direct phase mode.
\newblock In {\em Proc. SPIE}, volume 6272, pages 627228--627228--12, 2006.

\bibitem{Pinna_07}
E.~Pinna, F.~Quir\'{o}s-Pacheco, S.~Esposito, A.~Puglisi, and P.~Stefanini.
\newblock Signal spatial filtering for co-phasing in seeing-limited conditions.
\newblock {\em Optics Letters}, 32(23):3465--3467, 12 2007.

\bibitem{QuPa10}
F.~Quir\'{o}s-Pacheco, C.~Correia, and S.~Esposito.
\newblock Fourier transform wavefront reconstruction for the pyramid wavefront
  sensor.
\newblock In {\em 1st AO4ELT Conference - Adaptative Optics for Extremely Large
  Telescopes proceedings}, volume 07005, pages 1--6. EDP Sciences, 2010.

\bibitem{Raga96}
R.~Ragazzoni.
\newblock Pupil plane wavefront sensing with an oscillating prism.
\newblock {\em J. of Modern Optics}, 43(2):289--293, 1996.

\bibitem{RaDi02}
R.~Ragazzoni, E.~Diolaiti, and E.~Vernet.
\newblock A pyramid wavefront sensor with no dynamic modulation.
\newblock {\em Optics Communications}, 208:51--60, 2002.

\bibitem{RaTesch04}
R.~Ramlau and G.~Teschke.
\newblock Regularization of sobolev embedding operators and applications part
  i: Fourier and wavelet based methods.
\newblock {\em Sampling Theory in Signal and Image Processing}, 3(2):175--196,
  2004.

\bibitem{Ray93}
M.~Rayden.
\newblock On the barzilai and borwein choice of steplength for the gradient
  method.
\newblock {\em IMA Journal of Numerical Analysis}, 13(3):321--326, 1993.

\bibitem{Rigaut_2013_ao4elt3_yao}
F.~Rigaut and M.~Van~Dam.
\newblock {Simulating astronomical Adaptive Optics systems using yao}.
\newblock Proc. AO4ELT3, page~18, 2013.

\bibitem{Ros11}
M.~Rosensteiner.
\newblock Cumulative reconstructor: fast wavefront reconstruction algorithm for
  {E}xtremely {L}arge {T}elescopes.
\newblock {\em J. Opt. Soc. Am. A}, 28(10):2132--2138, 10 2011.

\bibitem{Ros12}
M.~Rosensteiner.
\newblock Wavefront reconstruction for extremely large telescopes via {CuRe}
  with domain decomposition.
\newblock {\em J. Opt. Soc. Am. A}, 29(11):2328--2336, 11 2012.

\bibitem{Sa16}
D.~Saxenhuber.
\newblock {\em Gradient-based reconstruction algorithms for atmospheric
  tomography in {A}daptive {O}ptics systems for {E}xtremely {L}arge
  {T}elescopes}.
\newblock PhD thesis, Johannes Kepler University Linz, 2016.

\bibitem{SaRa15}
D.~Saxenhuber and R.~Ramlau.
\newblock A gradient-based method for atmospheric tomography.
\newblock {\em Inverse Problems and Imaging}, 10(3):781--805, 2016.

\bibitem{schwartz_spie2018}
N.~Schwartz, J.-F. Sauvage, C.~Correia, B.~Neichel, T.~Fusco,
  F.~Quiros-Pacheco, K.~Dohlen, K.~El~Hadi, G.~Agapito, N.~Thatte, and
  F.~Clarke.
\newblock Analysis and mitigation of pupil discontinuities on adaptive optics
  performance.
\newblock In {\em Proc. SPIE}, volume 10703, pages 10703 -- 1070322 -- 12,
  2018.

\bibitem{schwartz_ao4elt5}
N.~Schwartz, J.-F. Sauvage, C.~Correia, C.~Petit, F.~Quiros-Pacheco, T.~Fusco,
  K.~Dohlen, K.~El~Hadi, N.~Thatte, F.~Clarke, J.~Paufique, and J.~Vernet.
\newblock {Sensing and control of segmented mirrors with a pyramid wavefront
  sensor in the presence of spiders}.
\newblock In {\em Proceedings AO4ELT5}, 2017.

\bibitem{Shatokhina_PhDThesis}
Iu. Shatokhina.
\newblock {\em {Fast wavefront reconstruction algorithms for extreme adaptive
  optics}}.
\newblock PhD thesis, Johannes Kepler University Linz, 2014.

\bibitem{Shat17_ao4elt5_clif}
Iu. Shatokhina, V.~Hutterer, and R.~Ramlau.
\newblock {Two novel algorithms for wavefront reconstruction from pyramid
  sensor data: Convolution with Linearized Inverse Filter and Pyramid Fourier
  Transform Reconstructor}.
\newblock In {\em Proceedings AO4ELT5}, 2017.

\bibitem{Shat_SPIE}
Iu. Shatokhina, A.~Obereder, and R.~Ramlau.
\newblock Fast algorithm for wavefront reconstruction in {XAO/SCAO} with
  pyramid wavefront sensor.
\newblock In {\em Adaptive Optics Systems IV}, Proc. SPIE 9148, pages
  91480P--1--15, 2014.

\bibitem{Shat13}
Iu. Shatokhina, A.~Obereder, M.~Rosensteiner, and R.~Ramlau.
\newblock {Preprocessed cumulative reconstructor with domain decomposition: a
  fast wavefront reconstruction method for pyramid wavefront sensor}.
\newblock {\em {Applied Optics}}, 52(12):2640--2652, 2013.

\bibitem{Shat17}
Iu. Shatokhina and R.~Ramlau.
\newblock {Convolution- and Fourier-transform-based reconstructors for pyramid
  wavefront sensor}.
\newblock {\em {Applied Optics}}, 56(22):6381--6390, 2017.

\bibitem{VanDam_2012_pwfs_truth_LTAO_GMT}
M.~A. van Dam, R.~Conan, A.~H. Bouchez, and B.~Espeland.
\newblock {Design of a truth sensor for the GMT laser tomography adaptive
  optics system}.
\newblock In {\em {Proc. SPIE 8447, Adaptive Optics Systems III}}, page 844717,
  2012.

\bibitem{Veran_ao4elt4_pwfs_vs_sh}
J.-P. V\'{e}ran, S.~Esposito, P.~Span\`{o}, G.~Herriot, and D.~Andersen.
\newblock {Pyramid versus Shack-Hartmann: Trade Study Results for the NFIRAOS
  NGS WFS}.
\newblock In {\em Proceedings of the Fourth AO4ELT Conference}, 2015.

\bibitem{Veri04}
C.~V\'{e}rinaud.
\newblock On the nature of the measurements provided by a pyramid wave-front
  sensor.
\newblock {\em Optics Communications}, 233:27--38, 2004.

\bibitem{Zhou06}
B.~Zhou, L.~Gao, and Y.-H. Dai.
\newblock Gradient methods with adaptive step-sizes.
\newblock {\em Computational Optimization and Applications}, 35(1):69--86,
  2006.

\end{thebibliography}

\end{document}